\documentclass[a4paper, authoryear, 10.5pt]{article}

\input epsf.sty
\usepackage{natbib}
\usepackage[english]{babel}
\usepackage{color}
\usepackage{babel}
\usepackage{enumitem}
\usepackage{fontenc}
\usepackage{graphicx}
\usepackage{amsmath,amsthm,amssymb}
\usepackage[colorlinks=true,citecolor=red]{hyperref}
\usepackage{lineno}
\usepackage[usenames,dvipsnames]{pstricks}
\usepackage{amsmath, amsfonts, amssymb, mathrsfs,hyperref}
\usepackage{epsfig}
\usepackage{amsthm}
\usepackage{gensymb}
\usepackage{float}
\begin{document}
\numberwithin{equation}{section}
\newtheorem{mydef}{Definition}
\newtheorem{remark}{Remark}
\newtheorem{thm}{Theorem}
\newtheorem{lemma}{Lemma}
\newtheorem{prop}{Proposition}
\newtheorem{corol}{Corollary}

\author{Anurag Sau \footnote{Agricultural and Ecological research Unit,ISI Kolkata,203,B.T Road,Kolkata 700108,India. E-mail:anuragsau@gmail.com}, Sabyasachi Bhattacharya  \footnote{Agricultural and Ecological Research Unit, Indian Statistical Institute, 203, B. T. Road, Kolkata, 700108, India. E-mail: sabyasachi@isical.ac.in},  Bapi Saha \footnote{Government College of Engineering \& Textile Technology, 4 Barrack Square, Berhampore, Email: bapi.math@gmail.com}, \footnote{Corresponding Author, Contact No. - (+91)-9339831862}}

\title{Recognizing and prevention of probable regime shift in density regulated and Allee type stochastic harvesting model with application to herring conservation  }
\maketitle

\begin{abstract}
	
%An ecological system having multiple stable equilibria is prone to undergo catastrophic shift from one stable state to another. If one of the stable states is an extinction state, the catastrophic change may lead to extinction. We consider two stochastic models of density regulation or Allee type density regulated models [Saha et al., Ecological Modelling, 2013] with linear and nonlinear harvesting terms in a closed system. The demographic noise is incorporated in terms of the birth-death process. We observe that, the catastrophic changes can be avoided by suitably choosing the handling time that will eventually help to prevent the sudden extinction of the harvested population. Furthermore, an ecological system having a single stable state in the deterministic setup may exhibit multiple stable states in the stochastic framework citing the density regulated growth laws. The study is illustrated on the herring population data based on the Global Population Dynamics Database.
An ecological system with multiple stable equilibria is prone to undergo catastrophic change or regime shift from one steady-state to another. It should be noted that, if one of the steady states is an extinction state, the catastrophic change may lead to
extinction. A suitable manual measure may control the prevention of catastrophic changes of different species from one equilibrium to another. We consider two stochastic models with linear and nonlinear harvesting terms. We inspect either density
regulation or Allee type density regulated models [Saha et al., Ecological Modelling, 2013], which have substantial applications in the herring fish population's viability study. Both the deterministic models we consider here contain bi-stability under certain restrictions, and in that case, one of the stable states is the extinction state. We assume that the dynamical system under consideration is closed, i.e., immigration and emigration are absent. The demographic noise is introduced in the system by substituting an ordinary differential equation with a stochastic differential equation model, where the birth and death rates of the deterministic process are used to obtain the instantaneous mean and variance in the stochastic differential equation. Our study reveals that, the catastrophic changes can be avoided manually by a suitable choice of handling time that will eventually help to prevent the sudden extinction of the
harvested population. The entire study is illustrated through the herring population size data obtained from the Global Population Dynamics Database (GPDD) and simulation experiment.

\end{abstract}

\noindent \subsubsection*{Keyword:} Handling time; Potential function; Stationary distribution; Theta-Logistic model; Early warning toolbox;

\section{Introduction}

The exploitation of biological renewable resources and harvesting of these resources are common practices in fisheries, forestry, etc. A suitable way is essential to carry out the adequate requirement of the human need for a long time.  The interaction between human activities and the ecosystem has a significant impact on nature at local to global scales \citep{cline2014early}. Indiscriminate harvesting of biological resources may hamper the sustainability of the different rare and endangered species. Hence, harvesting should be carried out scientifically to obtain the maximum benefit without hampering the ecological balance. Improper harvesting may also cause a catastrophic change in the population under consideration in the ecosystem. Many theoretical models suggest that the ecosystem may switch instantly from a stable state to an alternative stable state which is termed as regime shift \citep{crepin2012regime, scheffer2003catastrophic, scheffer2001catastrophic}. It may happen that if a regime shift occurs, it will be impossible to restore the stock to previous levels \citep{polovina2005climate}. Moreover, regime shift is an important feature of the physical environment of ecosystems and has the capability to make an effect on stock productivity. The shift refers to low frequency, high magnitude fluctuation in the marine ecosystems presuming changes in community composition, population abundances, and trophic structure. These changes happen in the abundance of both exploited and unexploited populations. An alternative stable state for regime shift is an important ecological phenomenon in the environment. When an alternative stable state occurs in the ecology, multiple stable equilibria exist in its underlying deterministic skeleton, and many exciting phenomena can occur \citep{may1977thresholds}. A suitable example of regime shift phenomena in ecology is lake eutrophication, where water is polluted by high concentrated nutrients such as nitrogen and phosphorus \citep{carpenter1997dystrophy, wang2018modeling}.

Generally, the Allee effect manifests a population bi-stability in a deterministic setup \citep{hilker2009allee}, whereas the theta-logistic model does not exhibit the same stability criteria. Although, in the case of the density regulated growth process, one can not rule out the possibility of the existence of multiple stable states in the presence of harvesting other than linear harvesting. Most management measures are directed at the individual stock of single species and do not take into account species interaction, such as prey-predator relationship  \citep{KAR2013134}. This motivates us to consider the important single-species models with density regulated parameters in the Allee type phenomena. A study on individual population dynamics by Sibly et al. \cite{sibly2005regulation} concludes that most of the density and per capita growth rate (pgr) relationship of different species is concave in nature. This relationship can be well defined by the density-dependent theta logistic model. In many cases, the rapid depletion of population size may reduce the species' fitness in a concave pattern with additional exposure to the Allee effect. Hence the density regulated Allee model is appropriate to explain such growth phenomena. Sau et al. extended the Allee model with density regulation around carrying capacity under harvesting. The author introduced this model in studying the extinction status of one of the economically beneficial fishes, viz. Atlantic herring \cite{sau2020extended}. We believe that, exploring the harvesting issue in the theta-logistic and Allee model must be an exciting research area to be explored as it can cover the growth profile of a wide range of species.

It is worthy to mention that, if the growth process of species follows the Allee mechanism, it enhances the morality of species \cite{chattopadhyay2016allee}. Incorporating the harvesting issue in this model can be interpreted as additional mortality. In addition to this, when the harvesting is incorporated into the system, it may act in favour of the regime shift from a higher stable state to a lower stable state. Harvesting is associated directly with the reproduction process of the species. If the reproduction success is more, the chance of effective harvesting is naturally increased. In many research articles on harvesting, authors introduced stochasticity through white noise instead of the birth-death process \citep{cooke1986one, houoptimal}. Harvesting is an important issue for the closed system as immigration and emigration do not occur directly in this case. Thus, extinction is a purely absorbing state under the birth-death set up, but this is not the case for white noise. Note that, in the second case, the population can recover from an extinction state which is not possible for a closed system.  This justifies our argument that the association of demographic stochasticity through the white noise is not pertinent to the harvesting model. This inconsistency can be avoided if the demographic stochasticity is incorporated through the birth-death process in place of any other commonly used noises.  Thus, in this work, we frame a suitable birth-death process using the deterministic model to incorporate the stochasticity in the system \citep{allen2003comparison, swift2002stochastic, tuckwell2018elementary}.

In this article, we consider two stochastic models viz. theta-logistic and Allee type density regulated models with linear and nonlinear harvesting terms. Both the linear and nonlinear harvesting strategies are considered, although the emphasis is given on nonlinear harvest rate as it is more realistic compared to the linear one \cite{KAR2013134}. The nonlinear harvest rate involves an additional parameter known as handling time. We can search for an optimum choice of density regulated parameter and the handling time so that we can restrict the species from undergoing regime shift from a higher stable state to a lower one (may be an extinction state). This leads the species under consideration to stay in a stable and safe mode instead of the risk of extinction. Similar stochastic models were also analyzed in Sau et al. \cite{sau2020extended}. The extract of the paper was mainly associated with the species' probability of extinction and the expected time to extinction. We believe that these two measures are insufficient to properly explain species' extinction threat since these two measures cannot capture the ensuing catastrophic change or possibility of regime shift, which may lead to species extinction. This may happen for a system having multiple stable equilibria. For example, suppose the estimated probability of extinction is found to be small or the expected time to extinction is large. It may happen if the species maintain its stable state at a moderately high equilibrium state. But this does not give us any clue about the possible regime shift to a lower stable equilibrium (may be an extinction state) if the system is perturbed by the noise of sufficiently large strength. This entire part was missing in the work of Sau et al.\cite{sau2020extended}.

We demonstrated the theoretical findings with the help of time-series data of the herring fish population, obtained from the Global Population Dynamic Database (henceforth GPDD, \url{http://www.sw.ic.ac.uk/cpb/cpb/dpdd.html}) with GPDD ID-1741 and GPDD ID-1772. We use the early warning signal toolbox in R to generate generic early warning signals of the time series data. The value of the control parameter, handling time, is estimated under the suggestion of the proper harvesting policy so that the possibility of extinction can be minimized. We derive the stationary distribution using the Kolmogorov forward equation. The infinitesimal mean and variance are obtained from the expressions of birth rate and death rate. In a stochastic environment, the sustainability of the resource population can be viewed with the help of its statistical distribution or by an effective potential function. The number of local minima of a potential function gives some significant input in the nature of stability and the possibility of the regime shift from one stable state to another stable state. 

The paper is organized as follows. The section \ref{The proposed model} contains some ideas of both the density regulated and density regulated Allee type model considering the linear and nonlinear harvesting. Then we consider our models in both deterministic and stochastic frameworks. In section \ref{Optimum harvesting policy and sustainability}, we provide the sketch of the stationary distribution for both linear and nonlinear harvesting. We examine sustainability through the potential function, regime shift phenomena, etc. In the next section \ref{Numerical simulations and results}, we simulate our result through numerical technique, and in section \ref{The real data and discussion}, we interpret our result through real-world data. Finally, we end with a conclusion in section \ref{Discussion and conclusion}.

%\textcolor{red}{In the previous article \citep{sau2020extended}, we find out the probability of extinction and expected time to extinction to find the possible extinction status of the herring fish population on some geographical location. The species following theta logistic model is found to be more stable in the sense of extinction criteria. We do not consider any stable region for species extinction. In this article, we analyze the stable zone of the species with the basin of attraction. If the species go to a near-zero stable equilibrium region from the nonzero stable equilibrium region (near carrying capacity), its probability of extinction will be high, and the species will go to extinction. In this article, we raise the flag of early warning to concern about the extinction phenomena of the species with proper harvesting strategy. This phenomenon motivates us to explain this article.    }

\section{Model formulation}\label{The proposed model}
\subsection{The deterministic model}  
The Allee effect is a biological phenomenon identified by a correlation between population density and the mean individual fitness of a population. An Allee effect may be pervasive when the population size becomes low, and a critical threshold value exists underneath the per-capita growth rate becomes negative. The population may be vulnerable to extinction due to this critical depensation level \cite{chattopadhyay2016allee}. For example, many aquatic animals such as fishes move from one region to another in a school, as schooling is generally considered a defense policy against predator attacks. Again the indiscriminate harvesting in dense school is a major reason for the decline of the population into danger \cite{courchamp2008allee}. In such cases, the strength of density dependence can have potential implications for the preservation of animal populations \citep{saha2013evidence}. Note that, apart from the Allee threshold, the strength of the density regulation parameter $\theta$ has an important role. To capture Allee phenomena together with density regulation around carrying capacity, we consider the equation according to Saha et al. \citep{saha2013evidence}   

\noindent\begin{equation}\label{general_Alleemodel}
\frac{dx}{dt}= rx\left(\frac{x}{K}-\frac{A}{K}\right)\left(1-\left(\frac{x}{K}\right)^\theta\right),
\end{equation}
where $x(t)$ is the biomass at time $t$, $r$ is called the intrinsic growth rate, $K$ is the carrying capacity, $\theta$ is called the density regulation around carrying capacity. We call this model as Allee-Saha model (ASM) \cite{sau2020extended}.

\noindent However, in reality, most of the species follow $\theta-$logistic process \cite{sibly2005regulation}, and hence we also consider the $\theta$-logistic growth process to encompass a large number of species. In the $\theta$-logistic model, the per-capita growth rate (pgr) decreases monotonically with the increase in population size. The theta logistic model is widely used to clarify the density dependence in the real population as well as the time-series data. The general theta-logistic model can be represented as \citep{bhowmick2016simple} 
\begin{eqnarray}\label{general_model} 
\frac{dx}{dt} = rx\left[1 - \left(\frac{x}{K} \right)^\theta \right].
\end{eqnarray}  
The parameters are described earlier. The pgr growth profile of the species is convex or concave depending on the value of $\theta$. The relationship is convex or concave according as the value of $\theta$ is greater or less than 1 respectively \citep{sibly2005regulation}. When $\theta$ is equal to 1, the equation \eqref{general_model} become classical logistic model.

The sustainability of different rare and endangered species may greatly depend on harvesting strategy \cite{sau2020extended}. We consider two types of harvesting strategies viz. proportional or linear harvest rate and nonlinear harvest rate. In general, the linear harvesting strategy is determined by the catch-per-unit effort hypothesis \citep{clark1990mathematical} and written as $H(x,E)=qEx$, where $q$ is the catchability coefficient and $E$ is the harvesting effort. Note that, this harvesting strategy has some defect in the real-life situation, which may be described as follows \cite{ghosh2014sustainable}: 
\noindent \begin{enumerate} 
	\item It considers the random search for every resource population. 
	\item It assumes that every individual to the resource population is equally likely to be captured. 
	\item $H$ is unbounded with respect to $E$ for a fixed $x$. 
	\item $H$ is unbounded with respect to $x$ for a fixed $E$. 
\end{enumerate} 
Handling time is one of the important parameters to avoid the above four unrealistic features. Hence, the functional form of the harvesting policy is written as $\displaystyle H(x, E)=\frac{qEx}{aE+Lx}$\ \citep{ghosh2014sustainable}, in which $a$ is the degree of competition among the boats, fishermen and other technology used in fishing \citep{agnew1979optimal}. L is the product of capture rate and handling time \citep{abrams2000nature}. If the capture rate is assumed to be constant, $ L $ is dependent only on handling time, and hence, $ L $ can be treated as a representative of handling time. This improved version of harvesting exhibit a saturation effect with respect to both resource size and effort. It is clearly observed that, $\displaystyle h \rightarrow \left(\frac{q}{a}\right)x$\ as $E \rightarrow \infty $ when $x$ is fixed and $\displaystyle h \rightarrow \left(\frac{q}{L}\right)E$\ as $x\rightarrow \infty $ when $L$ is fixed  \citep{ghosh2014sustainable}.   
\noindent Henceforth, we call this modified harvest rate as non-proportional or nonlinear harvest rate. Introducing the harvesting phenomena in the models \eqref{general_Alleemodel} and \eqref{general_model}, we have 
\begin{eqnarray} 
\frac{dx}{dt}&=&rx\left(\frac{x}{K}-\frac{A}{K}\right)\left(1-\left(\frac{x}{K}\right)^\theta\right)-H(x),\label{general_Alleemodel_harvest}\\
\frac{dx}{dt}&=&rx\left[1 - \left(\frac{x}{K} \right)^\theta \right]-H(x),\label{general_model_harvest}
\end{eqnarray} 
where $H(x)$ is the harvesting component. In case of our model, we consider both the linear and nonlinear harvesting.

\subsection{Stochastic model}
In many cases, the system of Ito stochastic differential equations are used to study a random dynamical system. A commonly used technique to develop SDE is by studying possible changes in the system components during a small interval of time. In this case, the transitions in the system are studied over a small time interval, and a differential equation is formed by considering the time interval approaches zero \citep{ allen2008construction, allen2003comparison}. As an example, the birth-death process may be applicable to introduce stochasticity in the system.

A stochastic process is addressed as a birth-death process where jumps from a particular state (number of individuals, lineages, cells, etc.) are only allowed to neighboring states. If the number of individual or similar quantities is increased by one, it represents the birth, while a decline to the left-hand part constitutes death. These characteristics of the birth-death process help us clarify the mathematical analysis, but we apply this process in different real-world mathematical modeling. This type of model permits us to deal with any questions regarding the transition or state probabilities of the process, mean, variance, stationary distribution, the meantime of existence, probabilities of extinction, etc. The outcomes procure with these models can be compared with real-world data.  
 
The presence of the Allee effect invites the possibility of regime shift, which may lead to the extinction of species. Harvesting is associated directly with the reproduction process of the species. Demographic stochasticity is an inherent property of the species which is directly related to the reproduction success. In comparison, environmental stochasticity is a sudden event that affects species reproduction by certain chance factors. In the present investigation, we have focused on the effect of demographic stochasticity in the system. Environmental stochasticity, where random fluctuation in the environment is effective to the entire population, is not considered here.

Introducing demographic stochasticity through white noise has certain limitations. This type of stochastic perturbation biologically represents disturbances, which occurs independently of current population densities but depends on the immigration and emigration of the species \citep{abbott2016alternative}. Harvesting is such a system where immigration and emigration do not occur directly in the model. So the incorporation of demographic stochasticity through white noise is not appropriate. This complication can be avoided if the demographic stochasticity is incorporated through the birth-death process instead of any other commonly used demographic noise \citep{allen2003comparison, swift2002stochastic, tuckwell2018elementary}. In our study, both the time and state space are continuous variables. Let $p(x,t)$ indicates the probability density function (p.d.f)for the continuous random variable $X(t)$.
Then, the mathematical notation of p.d.f is written as $\int_{a}^{b} p(x,t) dx$ = Prob$\{X(t)\in [a,b]\}$ .  

It is a well known established formula \citep{allen2003comparison} that, the probability density function for the general birth and death process satisfies the Kolmogorov forward differential equation 
\begin{equation}\label{Kolmogorov} 
\frac{\partial p(x,t)}{\partial t}=-\frac{\partial[(b(x)-d(x)) p(x,t)]}{\partial x}+\frac{1}{2}\frac{\partial^2 [(b(x)+d(x))p(x,t)]}{\partial x^2} ,
\end{equation} 
\noindent where $X(t)\in (0,N), \ t\in(0,\infty), \ p(x,0)=\delta(x-x_0)$, $b(x)$ and $d(x)$ are the birth and death rates. Here  $\mu(x)$ = $b(x)-d(x)$ is called the infinitesimal mean and $\sigma^2(x)$= $b(x)+d(x)$ is called infinitesimal variance.
 
\noindent The sample paths, $X(t)$ of a stochastic process satisfy the succeeding Ito stochastic integral equation \citep{allen2003comparison,gardiner1985input,wissel1987avoid} 
\begin{eqnarray}\label{general} 
X(t)= x_0 + \int_{0}^{t}[b(X(u))-d(X(u))]du+\int_{0}^{t}\sqrt{b(X(u))+d(X(u))}~dW(u), \ X(0)=x_0>0.
\end{eqnarray}   

\noindent Here, the first integral part of the \eqref{general} is the deterministic part, i.e., Riemann integrable, and the second integral part is a stochastic Ito integral part. Here $W$ is known as Wiener process and defined as $\triangle W(t)=W(t+\triangle t)-W(t)$ and it satisfies normal distribution, $N(0,\triangle t)$. For simplicity, this equation can be expressed as the SDE \citep{allen2003comparison} 
\begin{equation}\label{stochastic_birth} 
\frac{dX(t)}{dt}=b(X(t))-d(X(t))+\sqrt{b(X(t))+d(X(t))}~ \frac{dW}{dt}. 
\end{equation} 

In a population model, the birth and death function should follow the following properties. For ASM, we assume that there exist real numbers $M$, $N$ and $A^*$ such that,
\begin{enumerate}
	\item $b(0) = d(0) = 0$ and $b(x) \leqslant 0$ for $x \geq N$,
	\item $b(x) > 0$ for $x\in (0, N]$ and $d(x) > 0$ for $x \in (0, N]$,
	\item $b(x) < d(x)$ for $x \in [0,A^*)$ and $ (M, N)$,
	\item There exists a threshold value $A^*$ such that $b(x) > d(x)$ for $x \in (A^*, M]$.
\end{enumerate}
In addition to the above four properties $b(x)$ and $d(x)$ are $C^2([0,N])$.

For $\theta$-logistic model, we assume that there exist two real numbers $M$ and $N$ with $M<N$ such that \citep{allen2003comparison} 
\begin{enumerate}
	\item $b(0) = d(0) = 0$ and $b(x) \leqslant 0$ for $x \geq N$.
	\item $b(x) > 0$ for $x\in (0, N]$ and $d(x) > 0$ for $x \in (0, N]$.
	\item $b(x) > d(x)$ for $x \in (0, M]$.
	\item $b(x) < d(x)$ for $x \in (M, N)$.
\end{enumerate}
In addition to the above four properties, $b(x)$ and $d(x)$ are $C^2([0,N])$.

The feasibility of these regularity conditions is well discussed in Appendix A.

There are multiple ways to split up the single growth rate into the difference of a birth and death rate for the $\theta$-logistic and ASM. Among many considerable forms, we have chosen the most ecological meaningful choice for our model considered here. Basically, in the $\theta$-logistic model, the per capita growth rate (pgr) is high at low population density, and it is maximum at zero population size. The pgr diminishes with the increase in population size as more populations are competing for the same limited resources. Hence, the crowding effect, which is commonly considered to be the reflection of competition is present in the species which follows the $\theta$-logistic growth process. 
On the other hand, a species the growth process of which follows the Allee mechanism generally adopt the cooperation at low density to avoid becoming extinction \cite{abbasi2019}. Thus the competition or, in other words, the crowding effect can be avoided in the Allee model.

After careful consideration of these facts and the above four properties, we obtain the expressions for the birth and death rates corresponding to \eqref{general_Alleemodel_harvest} and \eqref{general_model_harvest} according to \cite{kang2014dynamics}. We obtain the following expressions for birth and death rates respectively,

\begin{eqnarray}\label{birth-death-rate-Allee} 
b(x)=\frac{rx^2}{K^{\theta+1}}\left(K^\theta+Ax^{\theta-1}\right)\left[1-\frac{x^\theta}{K^\theta+Ax^{\theta-1}}\right], \ d(x)=\frac{rax}{K}+H(x)
\end{eqnarray}
 and
\begin{eqnarray}\label{birth_death_thetalogistic} 
b(x)= (r+1)x-\frac{x^{\theta+1}}{2K^\theta}, \ d(x)= x+\frac{x^{\theta+1}}{2K^\theta} + H(x) .
\end{eqnarray}

\noindent Using two sets of birth and death rates, the ASM and the $\theta$-logistic model with harvesting term are governed by the following stochastic differential equations

 {\scriptsize
	\begin{equation}\label{thetalogisticallee_stochastic_birth}
	dX = rX\left(\frac{X}{K}-\frac{A}{K}\right)\left[1-\left(\frac{X}{K}\right)^\theta\right] dt + \sqrt{\left[\frac{rX^2}{K^{\theta+1}}\left(K^\theta+AX^{\theta-1}\right)\left(1-\frac{X^\theta}{K^\theta+AX^{\theta-1}}\right)+\frac{raX}{K}+H(X)\right]}~dW(t)
	\end{equation}}

and

\begin{equation}\label{thetalogistic_stochastic_birth}
dX = rX\left[1 - \left(\frac{X}{K} \right)^\theta \right] dt +\sqrt{\left[(r+1)X-\frac{X^{\theta+1}}{2K^\theta}\right]+\left[X+\frac{X^{\theta+1}}{2K^\theta}+H(X)\right]}~dW(t).
\end{equation}

\noindent For the stochastic differential equations \ref{thetalogisticallee_stochastic_birth} and \ref{thetalogistic_stochastic_birth}, $X=0$ is an exit boundary point \cite{karlin}. It implies that $X(t)$ converges in finite time to the absorbing state $0$. Thus the only stationary distribution is the absorbing state at $0$ with respect to the Dirac delta measure. 

\subsection{Existance of bistability under nonlinear harvesting:}\label{stability_analysis|section}
The existence of the bistability of ASM is already discussed in \citep{sau2020extended}. Here $0$ is always a stable equilibrium point, and the nontrivial solutions are stable or unstable depends on certain conditions. But for the $\theta$-logistic model, the trivial solution is not always stable, and it also depends on some conditions.   
We have 
\begin{equation*}
\frac{dx}{dt}= rx\left(1-\left(\frac{x}{K}\right)^\theta\right)-\frac{qEx}{aE+Lx}
\end{equation*} 
The linearization around $K$ will yield the following equations
\begin{equation}
\frac{dx}{dt}=r\theta x\left (1-\frac{x}{K}\right)-\frac{qEx}{aE+Lx}
\end{equation}
The equilibrium points of the above system are $0$, $x_1$ and $x_2$ where $x_1$, $x_2$ are given by

%r\left (1-\frac{x}{K}\right)(aE+Lx)-qE=0
%\implies \frac{rL}{K}x^2-r\left (L-\frac{aE}{K} \right)x+E(q-ar)=0
%\end{eqnarray*}

$x_1 =\frac {r\theta \left(L-\frac{aE}{K}\right)+ \sqrt{r^2\theta^2(\left(L-\frac{aE}{K}\right)^2-\frac{4rLE}{K}(q-ra)}}{\frac{2r\theta L}{K}}$ and 

$x_2 =\frac {r\theta \left(L-\frac{aE}{K}\right)- \sqrt{r^2\theta^2(\left(L-\frac{aE}{K}\right)^2-\frac{4rLE}{K}(q-ra)}}{\frac{2r\theta L}{K}}$

If $(q-ra)>0$ and $\left(L-\frac{aE}{K}\right) >0$ then $x_2$ exists. But when $(q-ra)>0$, $0$ is stable \cite{sau2020extended}. This shows that when $(q-ra) >0$, the two stable equilibrium $0$ and $x_1$ is separated by the unstable equilibrium point $x_2$. On the other hand if $(q-ra)<0$, $x_2$ does not exist and $0$ becomes unstable.

\section{Stationary distribution and possibility of regime shift}\label{Optimum harvesting policy and sustainability}

Stationary distribution is one of the important measures to visualize the model’s behavior in a long time. It is the probability distribution that remains unchanged as time progress. The potential function is an alternative measurement of the stationary distribution, which is equally capable of determining the possibility of the regime shift of a system. The increase in the asymmetry in the potential function is the indication of catastrophic change or regime shift \cite{guttal2008changing, sharma2015stochasticity}. Moreover, the local minima of a potential function indicate the location of the stable equilibrium points of the deterministic counterpart of the model.

\subsection{Stationary distribution: $\theta$-logistic model}

\subsubsection*{Case 1: Linear harvest rate}

Consider $\{X(t), t\geqslant0\}$ be a continuous stochastic variable with transition density $p(t,x,y)$ for $t\geqslant0$, which satisfies the  Kolmogorov forward equation 
\begin{equation}\label{forward_kolmogorov} 
\frac{\partial p}{\partial t}=-\frac{\partial [p(t,x,y)(\mu(x))]}{\partial x}+\frac{1}{2}\frac{\partial^2 [(p(t,x,y)\sigma^2(x))]}{\partial x^2}. 
\end{equation} 
Where $\mu(x)$ is the drift coefficient or infinitesimal mean and $\sigma^2(x)$ is infinitesimal variance. Hence, there exist stationary density $f(.)$ (may or may not exist), which satisfy the following condition 
$\displaystyle f(x)=\int f(y)p(t,x,y)~dy$ for all $t\geqslant0.\ $\ 

In the case of $\theta$-logistic harvesting model \eqref{general_model_harvest}, where the parameters are represented as earlier, the stochastic differential equation corresponding to the deterministic model \eqref{general_model_harvest} is given by\eqref{thetalogistic_stochastic_birth}. We find the stationary distribution for the stationary state by the Kolmogorov forward equation and equating 
$\displaystyle \frac{\partial p}{\partial t}=0$.\  \\

\noindent We define the probability density function at stationary state as $P_s(x)$. 
In the case of the model  \eqref{Kolmogorov}, 
$\mu(x)=[b(x)-d(x)]$ and 
$\sigma^2(x)=[b(x)+d(x)].$

\noindent Now,\\
\begin{eqnarray*}
&& \frac{\partial P_s(x)}{\partial t}= 0\\
&&\Rightarrow\frac{d}{dx}(P_s(x)\mu(x))=\frac{1}{2}\frac{d^2}{dx^2}(P_s(x)\sigma^2(x))\\
&&\Rightarrow P_s(x)\mu(x)= \frac{1}{2}\frac{d}{dx}[P_s(x)\sigma^2(x)]+C \\
&&\Rightarrow P_s(x) = Nexp(-2U(x)) \mbox{, Assuming $C=0$.}
\end{eqnarray*}

\noindent Where $N$ is the normalizing constant and $\displaystyle U(x)=\int\frac{\sigma(x)\frac{d\sigma}{dx}-\mu(x)}{\sigma^2(x)}dx$\ \label{potential_all}, is the effective potential function \citep{guttal2008changing}. In particular for model \eqref{thetalogistic_stochastic_birth}, with $H(x,E)=h(E)x$,  
\begin{equation*}
P_s(x) =  N\left[\frac{1}{x}\exp\left(\psi(x)\right)\right],
\end{equation*}
where,
$N(E)$ is the normalizing constant and defined as
\begin{equation*}
N(E) = \left[\int_{1}^ \infty \frac{1}{x} \exp(\psi(x,E))dx\right]^{-1},
\end{equation*}
\noindent and \\
\begin{equation}
\psi(x,E) = \frac{2(r-h)}{r+2+h}{(x-x_0)} - \frac{2r}{(r+2+h)(\theta+1)K^\theta}(x^{\theta+1} - x_0^{\theta+1}).
\end{equation}
%\noindent The expected sustainable yield can be obtained as 
%\begin{eqnarray*}
%s(E)=E(H(X,E))=h(E)\int_{x_0}^ \infty xP_s(x)dx~~
%= h(E)N(E)\int_{x_0}^{\infty}\exp(\psi(x))dx
%\end{eqnarray*}

\subsubsection*{Case 2: Nonlinear harvest rate}

In the previous section, we give the expression of stationary distribution when the harvest is linear. Traditionally, the harvesting policy is generalized by the catch-per-unit effort hypothesis \citep{clark1990mathematical}. So the handling time is sensible for the bounded functional response. In this section, we consider the nonlinear harvesting function as $\displaystyle ~H(x,E) = \frac{qEx}{aE+Lx}$.\ \\
Similar by previous technique, the stationary probability density function is given by,
\begin{equation}\label{potential_function}
P_s(x)=N \exp\left[-2 U(x)\right].
\end{equation}
\noindent Here $N$ is the normalizing constant and $U(x)$ is called the effective potential function, which is defined as \citep{guttal2008changing}
\begin{equation}\label{Potential}
 U(x)=-\left[\int_{x_0}^{x} \frac{\mu(x)- \sigma(x) \sigma'(x)}{\sigma^2(x)}dx\right].
\end{equation}
Infinitesimal variance $\sigma^2(x)$ is defined in the previous section and $\displaystyle \sigma'(x)=\frac{d\sigma}{dx}$.\ In this case, the $\theta$-logistic model with nonlinear harvest rate takes the form, 
\begin{eqnarray} \label{birth_death}
\frac{dx}{dt}&=&rx\left[1-\left(\frac{x}{K}\right)^\theta\right]-\frac{qEx}{aE+Lx} \nonumber \\  
&=&\left[(r+1)x-\frac{x^{\theta+1}}{2K^\theta}\right]-\left[\left(x+\frac{1}{2}\frac{x^{\theta+1}}{K^\theta}\right)+\frac{qEx}{aE+Lx}\right].\nonumber\\
\end{eqnarray}
From the equation \eqref{birth_death}, we can write the birth and death terms as:
\begin{eqnarray*}
  b(x)= (r+1)x-\frac{x^{\theta+1}}{2K^\theta}~,~~
  d(x)= \left(x+\frac{x^{\theta+1}}{2K^\theta}\right)+\frac{qEx}{aE+Lx}.
  \end{eqnarray*}
  So,
  \begin{eqnarray*}\nonumber\label{Fokkerplank}
  b(x)+d(x)=\sigma^2(x)=\left[(r+2)x+\frac{qEx}{aE+Lx}\right]\\
  \Rightarrow \sigma(x)\sigma'(x)= \frac{1}{2}\left[(r+2)+\frac{aqE^2}{(aE+Lx)^2}\right]
  \end{eqnarray*}
\noindent The Kolmogorov forward equation in this case can be written as,
\begin{equation}
\frac{\partial p}{\partial t}=-\frac{\partial}{\partial x}\left[\left(rx\left[1-\left(\frac{x}{K}\right)^\theta\right]-\frac{qEx}{aE+Lx}\right) p(x,t)\right]+\frac{1}{2}\frac{\partial^2}{\partial x^2}\left[\left((r+2)x+\frac{qEx}{aE+Lx}\right)p(x,t)\right]\nonumber
\end{equation}

\noindent So from \eqref{Potential} we obtain the following expression,\\
\begin{align*}
U(x) &=-\left[\int_{x_0}^{x}\frac{rx\left(1-(\frac{x}{K})^\theta\right)-\frac{qEx}{aE+Lx}-\frac{r+2}{2}-\frac{aqE^2}{2(aE+Lx)^2}}{(r+2)x+\frac{qEx}{aE+Lx}}\right]dx\\
&=-\left[\int_{x_0}^{x}\frac{rx(aE+Lx)^2\left(1-(\frac{x}{K})^\theta\right)-qEx(aE+Lx)-\frac{r+2}{2}(aE+Lx)^2-\frac{aqE^2}{2}}{(r+2)x(aE+Lx)^2+qEx(aE+Lx)}\right]dx
\end{align*}
\noindent The exact mathematical expression of $U(x)$ can not be established due to the presence of nonlinearity in the population biomass. To resolve the problem, we linearize the term $\left(1-(\frac{x}{K})^\theta\right)$. The approximation of $\left(1-(\frac{x}{K})^\theta\right)$ can be made using binomial expansion as follows:

\begin{align*}
1-\left(\frac{x}{k}\right)^\theta 
&= 1-\left(\frac{K-(K-x)}{K}\right)^\theta \\ 
&=1-\left[1-\theta\left(\frac{K-x}{K}\right)+\frac{\theta(\theta+1)}{2}\left(\frac{(K-x)}{K}\right)^2+\ldots\right] ; 0<x<2K\\
&=(\theta-\frac{\theta}{K}x).\\
\end{align*}

\noindent Substituting the value in the above expression of $U(x)$, we get,
\begin{align}
U(x) &= -\left[\int_{x_0}^{x}\frac{r\theta x (1-\frac{x}{K}) (aE+Lx)^2-qEx(aE+Lx)-\frac{r+2}{2}(aE+Lx)^2-\frac{aqE^2}{2}}{x(aE+Lx)[(r+2) (aE+Lx)+qE]}\right]dx\\ \nonumber
&= -\frac{r \theta}{K}\left[\frac{2(D(D-C)+BC^2)\log(Cx+D)+Cx(Cx+2AC-2D)}{2C^3} \right]+ \frac{1}{2}\log\left((r+2)x+\frac{qEx}{aE+Lx}\right)\\
& =\frac{r \theta L}{K}\left[\frac{2(D(D-C)+BC^2)\log(Cx+D)+Cx(Cx+2AC-2D)}{2C^3} \right]+\frac{1}{2} \log\left((r+2)x+\frac{qEx}{aE+Lx}\right) \nonumber
%&= -r\theta\int_{x_0}^{x}\frac{(aE+Lx)(1-\frac{x}{K})}{[qE+(r+2)(aE+Lx)]}dx+ qE\int_{x_0}^{x}\frac{1}{[qE+(r+2)(aE+Lx)]}dx\\ \nonumber
%&+ \left(\frac{r+2}{2}\right)\int_{x_0}^{x}\frac{(aE+Lx)}{x[qE+(r+2)(aE+Lx)]}dx+\frac{aqE^2}{2}\int_{x_0}^{x}\frac{dx}{x(aE+Lx)[qE+(r+2)(aE+Lx)]}\\  \nonumber
%&=-r\theta I_1+qEI_2+\left(\frac{r+2}{2}\right)I_3+\frac{aqE^2}{2}I_4 
\end{align}

Here $A= \frac{aE-KL}{L}$, $B=\frac{(q-a)KE}{L}$, $C=L(r+2)$, $D= qE+aE(r+2)$.
%\noindent Here $I_1$, $I_2$, $I_3$ and $I_4$ are given by ,\\
%\begin{eqnarray*}
%&&I_1=-\frac{1}{2k(2+r)^3}\left[-\frac{2(2+r)(1+\frac{L}{aE}x(2+r))aEx}{L}+(2+r)^2x^2+\frac{a^2E^2}{L^2}\left(\frac{2q^2}{a}\left(2+\frac{q^2}{a}+r\right)\right.\right.\\
%&&\left.\left.\hskip 1cm+\frac{L}{aE}x(2+r)\right)\log\left[\frac{q^2}{a}+(2+r)\left(1+\frac{L}{aE}x\right)\right]\right]\\
%&&I_2=\frac{aE\log\left[\frac{q^2}{a}+(2+r)(1+\frac{L}{aE}x)\right]}{(2+r)L}\\
%&&I_3=\frac{(2+r)\log x+\frac{q^2}{a} \log\left[\frac{q^2}{a}+(2+r)\left(1+\frac{L}{aE}x\right)\right]}{(2+r)\left(2+r+\frac{q^2}{a}\right)}\\
%&&I_4=\frac{a}{q^2\left(2+\frac{q^2}{a}+r\right)}{\left(\frac{q^2}{a} \log x\right)-\left(2+\frac{q^2}{a}+r\right)\log\left(1+\frac{L}{aE}x\right)+(2+r)\log\left[\frac{q^2}{a}+(2+r)\left(1+\frac{L}{aE}x\right)\right]}
%\end{eqnarray*}
Note that, this form of U(x) is valid for $0<x<2K$.

\subsection{Stationary distribution: ASM}\label{stationary distribution}

\subsubsection*{Linear and nonlinear harvest rate}

In this section, we will derive the stationary distribution under nonlinear harvest rate and the corresponding stationary distribution in case of linear harvest rate can be easily obtained as a special case of the former distribution with $L$=0. Here, we will give some outline or sketch of the stationary distribution for ASM. The derivation of the analytic expression will be difficult in this case, so we mainly emphasize the numerical technique to analyze the potential function.    
For the ASM with harvesting, the forms of birth and death terms are 

\noindent $\displaystyle b(x)= \frac{rx^2(K^\theta+Ax^{\theta-1})}{K^{\theta+1}}\left(1-\frac{x^\theta}{K^\theta+Ax^{\theta-1}}\right)$,~~~ 
$d(x)=\frac{rAx}{K}+H(x)$.\       \\             \label{Allee-potential}
\noindent The harvest term may be linear or non linear according to our choice. Considering nonlinear harvest rate, the mean and variance can be written as\\
 $\mu(x)=rx\left(\frac{x}{K}-\frac{A}{K}\right)\left[1-\left(\frac{x}{K}\right)^\theta\right]-\frac{qEx}{aE+Lx}$\\ and\\
$ \sigma^2(x)=\frac{r(K^\theta+Ax^{\theta-1})}{K^{\theta+1}}x^2\left(1-\frac{x^\theta}{K^\theta+Ax^{\theta-1}}\right)+\frac{rAx}{K}+\frac{qEx}{aE+Lx}$~~

$ \Rightarrow \sigma(x)\sigma'(x)= r\left[\left(\frac{2K^{\theta}x+A(\theta+1)x^\theta}{K^{\theta+1}}\right)\left(1-\frac{x^\theta}{K^\theta+Ax^{\theta-1}}\right)-\left(\frac{K^\theta x^2+Ax^{\theta+1}}{K^{\theta+1}}\right)\left(\frac{(K^\theta+Ax^{\theta-1})\theta x^{\theta-1}-A(\theta-1)x^{2(\theta-1)}}{(K^\theta+Ax^{\theta-1})^2}\right)\right]+\frac{rA}{K}+\frac{qE^2a}{(aE+Lx)^2}$.

\noindent Substituting these values in the equation \eqref{Potential}, we will obtain the potential function. This expression is not solvable analytically. So we will analyze this properly through numerical technique, which is discussed in section \ref{Numerical simulations and results}. The potential function and hence the stationary distribution in case of proportional harvesting can be achieved by assigning $L=0$.

\section{Numerical simulations and results}\label{Numerical simulations and results}

\subsection{ Stochastic perturbation on density regulated Allee type model}

\noindent In nature, the dynamical systems are dependent on external interference. Many inconveniences such as pest outbreaks, changes in weather conditions, rainfall, fire, etc., may be harmful to the populations in multiple ways, such as exterminate a portion of it. If the system has a single steady-state, the system will settle back to the same state after such a perturbation. Whereas, if the system has multiple stable states, such as the Allee effect, a sufficiently large disturbance of the ecological state may put the ecosystem into the basin of attraction of another stable state. The plausibility does not depend only on the perturbation but also on the size of the basin of attraction. In terms of stability landscapes, if the valley is small, a small perturbation may be sufficient to displace the ball far enough to push it over the hill, resulting in a shift to the alternative stable state. For example, in fisheries, increased harvesting on piscivores can result in a shift from high piscivores low-planktivore to the state low-piscivores high planktivore regime \citep{walters2001cultivation}. Several fisheries are suspected of having suffered this type of transition \citep{steele2004regime, barange2008regime}.

In general, when the growth process of the species follows ASM in the absence of harvesting, the species has two stable equilibrium points, 0 and K, separated by an unstable equilibrium point at Allee threshold $A$. Here we introduce the nonlinear harvesting strategy, which is obtained incorporating handling time during the harvesting process. 

%100 sample paths are generated for the stochastic system \ref{thetalogisticallee_stochastic_birth} to depict the figure \ref{diffLsameinitialcondition}. The values of the parameter are chosen from the literature. Primarily, we generate the sample paths for different values of the handling time L, keeping all other parameters fixed. When the handling time is zero (for the linear harvest rate), all the sample paths are converging to zero stable state i.e, if we consider the linear harvest rate, the chance of species extinction is pretty high. Sufficient increase in the handling time leads most of the sample paths to converge to the upper stable equilibrium (figure \ref{diffLsameinitialcondition}) and hence the species is less likely to become extinct. It is observed that the initial population size plays a vital role in the stochastic model. If the initial population size is pretty close to (if not in the basin of attraction of extinction state) the basin of attraction of the extinction state, regime shift is observed, and it remains although $L$ is increased (figure \ref{diffLdiffinitial}). On the contrary, if the initial population size lies near the basin of attraction of the upper stable state, regime shift does not occur beyond a certain value of the control parameter $L$, and the population remains at the stable state (figure \ref{diffLsameinitialcondition}). 
 Since the exact expression for non zero equilibrium points is intractable, we rely on the bifurcation diagram in figure \ref{bifurcation}. When the handling time is small, the extinction state is the only stable equilibrium point. However, when we increase the handling time, there exists another stable equilibrium point, which is near the carrying capacity i.e, two stable equilibrium is separated by an unstable equilibrium. This observation establishes the fact that, if handling time (L) lies below a certain threshold value corresponding to the bifurcation point (as indicated in figure \ref{bifurcation}), extinction is inevitable. This bifurcation figure is also discussed in the work of Sau et al. \cite{sau2020extended}. In the long run, the distribution of population size can be represented by the stationary distribution of the population size. However, the computation of stationary distribution for ASM is very complex, as we observe in section \ref{Optimum harvesting policy and sustainability}. To avoid this complexity, we compute the potential function of the population size, which possesses just an opposite behavioral characteristic of the stationary probability density function. The potential function attains minima where the stationary distribution has mode and vice-versa. 
In this case, we observe that if $L$ increases, the minimum value of the potential function at the lower equilibrium point remains the same, although it decreases at the higher equilibrium point. It indicates that the population is more likely to stay around a higher equilibrium point if $L$ increases (figure \ref{probability_density_function}(a)). In addition, an increase in the minimum value at a higher equilibrium point indicates a higher chance of regime shift from upper stable equilibrium to lower stable equilibrium, which is the extinction state. Furthermore, it is observed that if $L$ is increased, the basin of attraction for the upper stable equilibrium point increases ( figure \ref{Basin of attraction_Allee}). This shows that a proper handling time (sufficiently large) can prevent possible catastrophic change and thus restrict population extinction. This argument is validated by numerical simulation in data analysis.

A similar result can be observed for $\theta$ also. If $\theta$ increases the intra-species competition decreases (since $(x/K)^\theta$ decreases for $x<K$) and naturally, the population is more likely to stay at higher equilibrium density (see figure \ref{probability_density_function} (b)).

\begin{figure}[H]
	\begin{center}
		\includegraphics[height = 70mm, width =150mm]{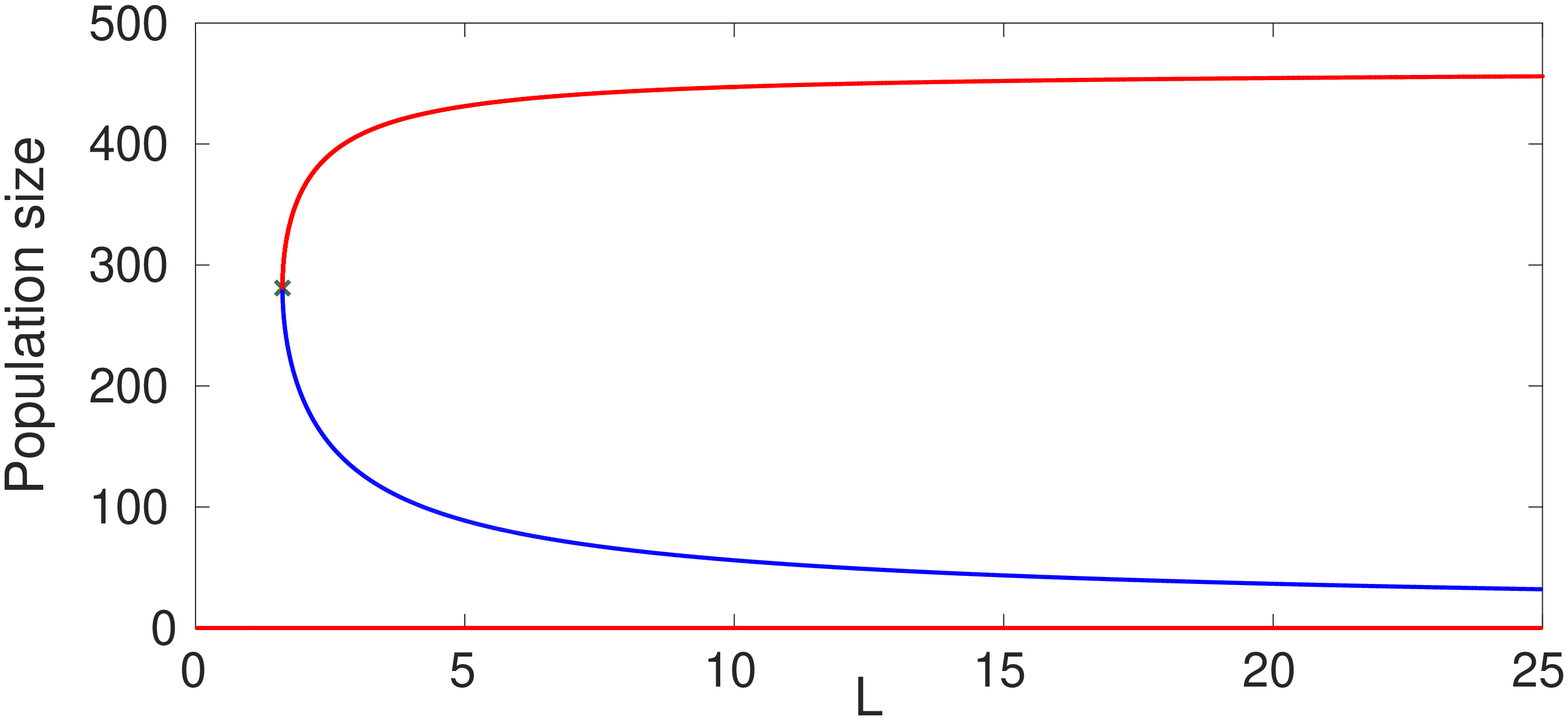}
	\end{center}
	\caption{ This figure represents the bifurcation diagram (saddle-node) of equilibrium points vs. $L$. The red dots represent the stable equilibrium points, and the blue dots represent unstable equilibrium points. The other parameters are $A=345$, $r=0.40$,
		$K=460$,
		$a=1.01$,
		$q=2.106$,
		$E=.446$,
		$\theta=0.017.$
	}
	\label{bifurcation}
\end{figure}

\begin{figure}[H]
	\begin{center}
		\includegraphics[height=2.2in,width=3in]{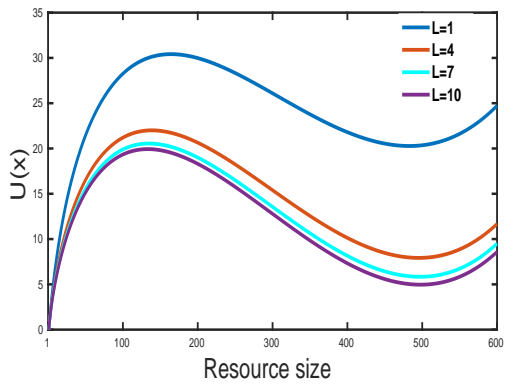} (a)
		\includegraphics[height=2.2in,width=3in]{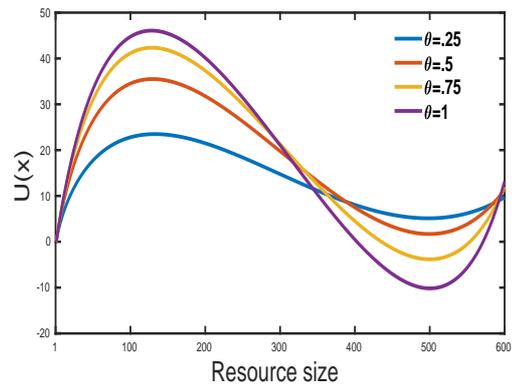}  (b)
		%\subfigure[]{\includegraphics[height=2.2in,width=3in]{pdfthetaallediffA.eps}}
	\end{center}
	\caption{The panels of figures exhibit the potential functions for different $\theta$ and $L$. In figure \ref{probability_density_function}(a), the set of parameter values are $r$ =0.407, $K$  = 500.484, $A$=80, $q$=2.10, $E$ =.446, $a$ =1.019, $L$=12. For figure \ref{probability_density_function}(b), the parameters are $r$ =0.407, $K$  = 500.484,  $\theta$ =.2, $A$=80, $q$=2.10, $E$ =.446, $a$ =1.019.}
	\label{probability_density_function}
\end{figure}

\begin{figure}[H]
	\begin{center}
		\includegraphics[height = 70mm, width =150mm]{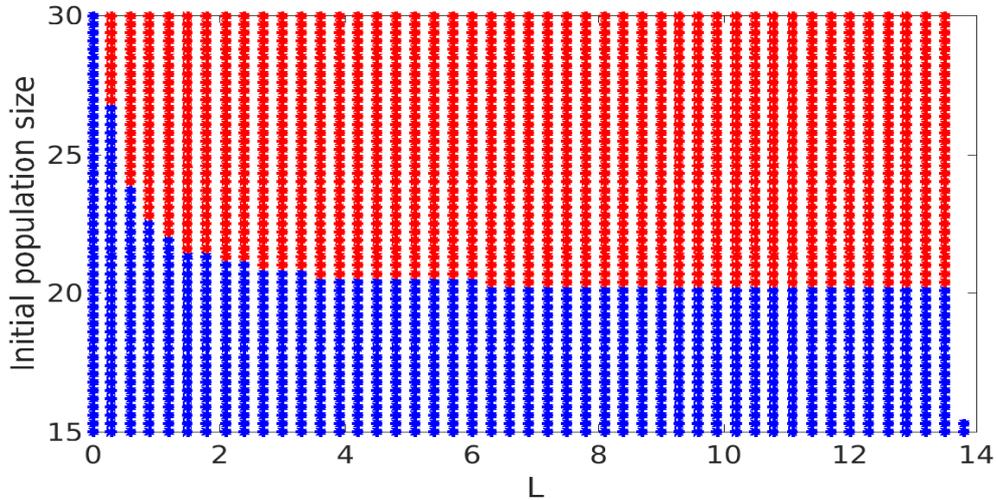}
	\end{center}
	\caption{This figure represents the basin of attraction for the Allee model with harvesting. The blue dots represent the basin of attraction for trivial stable equilibrium point, and the red dots represent that of the nonzero stable equilibrium. 
	}
	\label{Basin of attraction_Allee}
\end{figure}

\subsection{Density regulated harvesting model under stochastic environment}

%\textcolor{red}{The $\theta$-logistic model with nonlinear harvesting has multiple equilibrium points among which one is zero. In this case we observe in [sau et al], the stability analysis section that if$( r-\frac{q}{a}) <0$, there exist another two nonzero equilibrium points $x1$, $x2$ and $x1$ is stable. This shows that under the constraint $( r-\frac{q}{a}) <0$, the two stable equilibrium 0 and $x1$ are separated by an unstable equilibrium x2. So in this case maximum number of sample path will converges to $0$ if $L$ is small. Again if we increase the handling time, the basin of attraction near the non trivial equilibrium point will increase. Hence the maximum number of sample path will converge to the non trivial stable equilibrium point. Hence the possibility of species extinction will also reduce.  Again if $0$ is an unstable equilibrium then there must be a stable equilibrium in negative part of the axis. But for this situation our discussion is not necessary. This phenomena is depicted in the figure..........       }
  
In the previous section, we demonstrate the effect of different relevant parameters on the possible regime shift for ASM, which possesses bi-stability. In this section, we consider the $\theta$-logistic harvesting model, which may possess multiple stable equilibrium points under a nonlinear harvesting strategy. We also analyze the effect of different relevant parameters on the stationary density, which is obtained in the range of population size $[1,\infty)$. Figure \ref{pdfthetalogistic}(a) shows that, for linear harvest rate, the resource population becomes more likely to stay around the carrying capacity as $\theta$ increases. This phenomenon occurs because, increase in $\theta$ diminishes the intra-species competition, which in turn increases the species birth rate. It should be noted that, when $\theta$ is large, the species such as mammals generally likes to stay around carrying capacity and are more stable around it \citep{sibly2005regulation}. When a linear harvesting strategy is applied, for small values of $\theta$, the probability of staying near a low population size is high, which is alarming for the population. This problem can be overcome if nonlinear harvesting is applied. We observe that, if a suitable amount of handling time is applied, the probability of staying near population size 1 is negligible, although density regulation parameter $\theta$ is small (figure \ref{pdfthetalogistic}(b)). Furthermore, it can be observed (figure \ref{pdfthetalogistic} (c)) that if the handling time ($L$) is increased keeping all other parameters fixed, the population size has a tendency to stay near carrying capacity with high probability. Thus, we may enhance the possibility of survivability of species whose density regulation around carrying capacity is small by adopting nonlinear harvesting instead of linear. Density regulation around carrying capacity ($\theta$) is an inherent property of the species and for most of the species $\theta<1$ \citep{sibly2005regulation}. So we may conclude that nonlinear harvesting is beneficial for the sustainability of species.

A dynamical system faces a catastrophic shift when it has multiple stable equilibria separated by unstable equilibria. The theta-logistic model with nonlinear harvesting may possess multiple stable equilibria under certain choices of the parameters. To avoid mathematical complexity, we consider the linearized model as discussed in the section \ref{stability_analysis|section}. In this case, we observe that if $(q-ra) >0$, there exist another two nonzero equilibrium points $x_1$, $x_2$ and $x_1$ is stable. This shows that under the constraint $(q-ra) >0$, the two stable equilibrium $0$ and $x_1$ are separated by an unstable equilibrium $x_2$. The exact expression for the potential function is also obtained for this linearized model.

\noindent The potential function gradually loses its symmetry when the system enters into the unstable region from the stable region. The onset of asymmetry in the potential function indicates a high chance of catastrophic change in the system \citep{guttal2008changing}. This may lead the system to become extinct. The change in equilibrium density of the resource population with respect to $L$ and $\theta$ is depicted in figure \ref{equilibrium}, which is also given in \cite{sau2020extended}. The equilibrium density marked with red color represents the stable equilibrium, and those marked in blue color are considered as unstable points. As the choice of the parameters is shifted from the stable to the unstable region, the potential function loses its symmetry, which is evident in figure \ref{potentialfunction}. In addition, it can be seen in figure \ref{fig6} that an increase in L leads to an increase in the symmetry of the potential function, which establishes the fact that increased handling time reduces the chance of catastrophic regime shift. As the parameter theta is an inherent feature of the concerned species, we do not have any control over it. We can adjust handling time to avoid catastrophic change when a system contains bistable equilibrium points, and the species maintains its stable existence at the upper stable equilibrium point. We observe from figure \ref{Basin of attraction} that if L is increased, the basin of attraction of upper stable equilibrium also increases. This means that the population can withstand a larger depletion in population than that when L is small.   

 The linear (proportional) harvesting strategy is a special case of nonlinear harvesting when $L$ is chosen to be $0$. Hence, in every simulation activity, $L$ may be chosen to be $0$ to incorporate the proportional harvesting policy. It is observed that, when the parameters $L$, $\theta$ are chosen from the stable region (marked in red in figure \ref{equilibrium}), the population size is almost normally distributed with a mean near carrying capacity (figure \ref{stochastic_samplepath_histogram}(a)). If the chosen parameters are away from the stable region, the said distribution becomes asymmetric, and the clustering of population near $0$ increases (figure \ref{stochastic_samplepath_histogram}).

\begin{figure}[H]
	\begin{center}
		\includegraphics[height = 70mm, width =150mm]{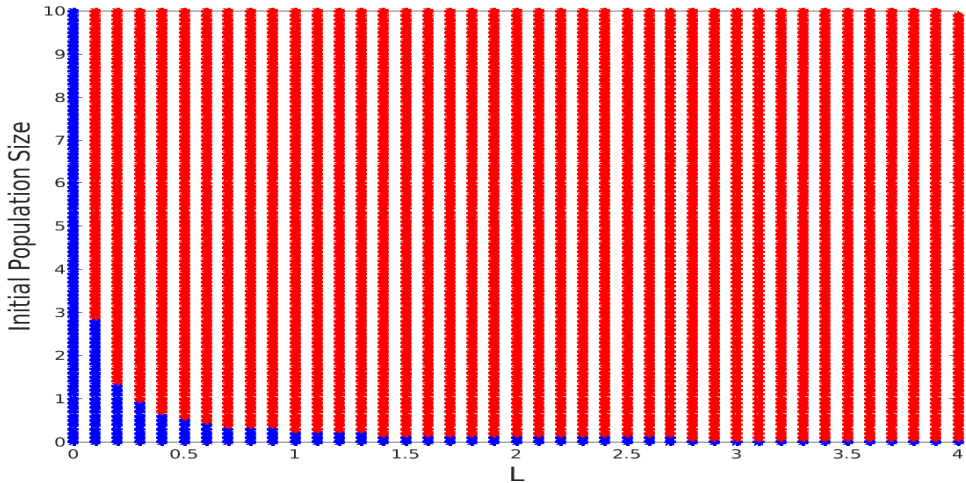}
	\end{center}
	\caption{ This figure represents the basin of attraction for the $\theta$-logistic model with harvesting. The blue dots represent the basin of attraction for trivial stable equilibrium point, and the red dots represent that of the nonzero stable equilibrium.
	}
	\label{Basin of attraction}
\end{figure}

\begin{figure}[H]
\begin{center}
\includegraphics[height = 55mm, width =65mm]{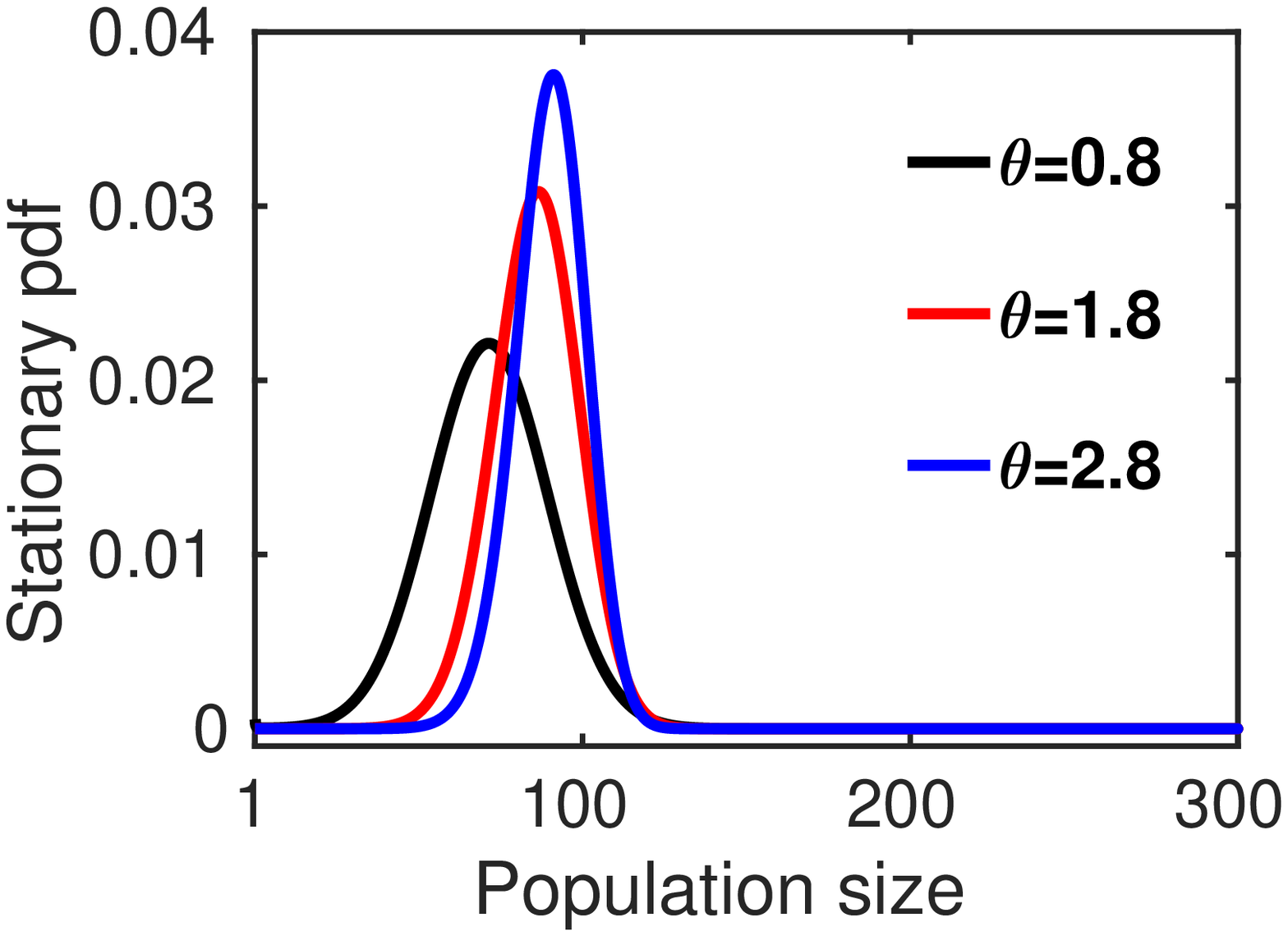}(a)
\includegraphics[height = 53mm, width = 65mm]{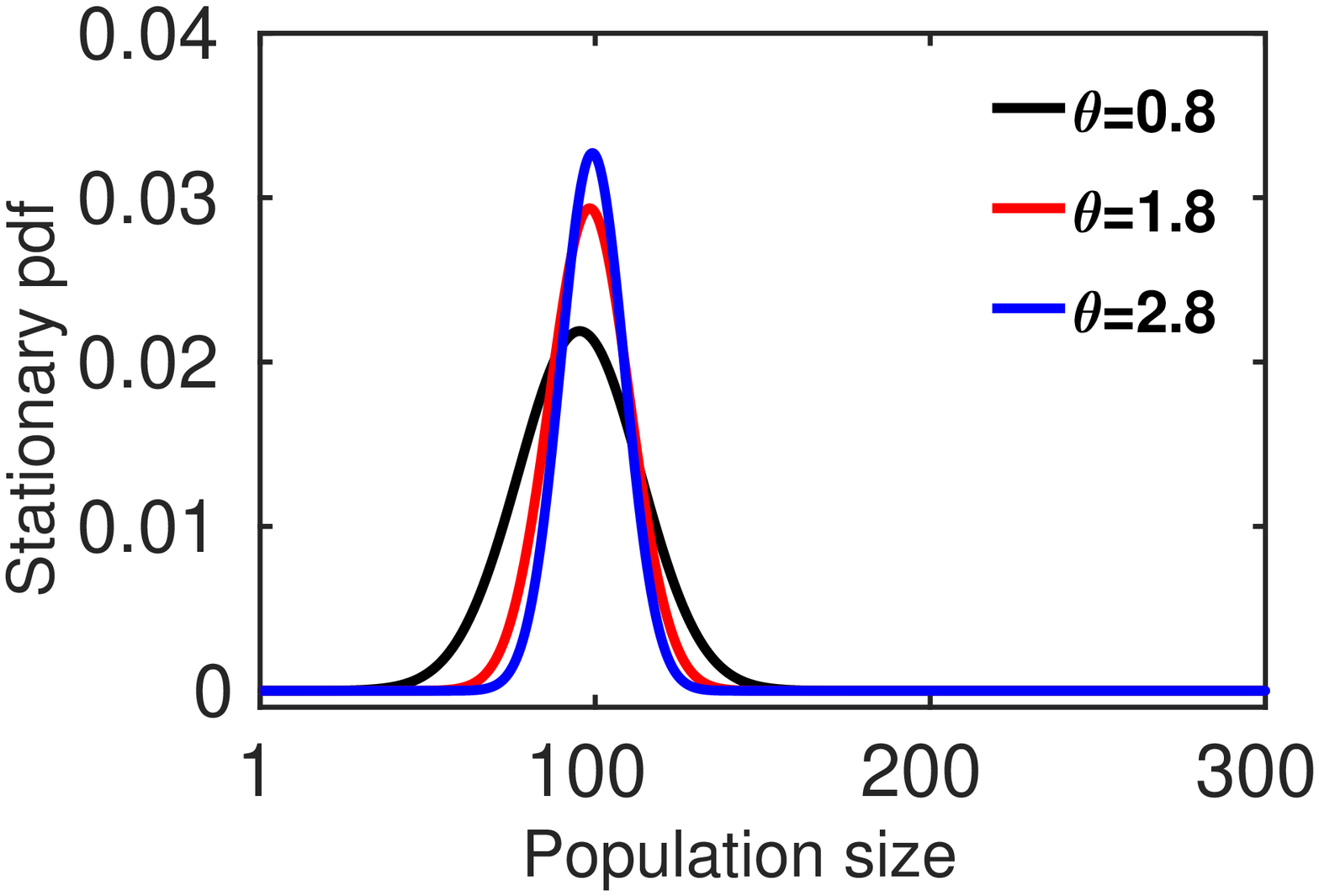}(b)
\includegraphics[height = 55mm, width =65mm]{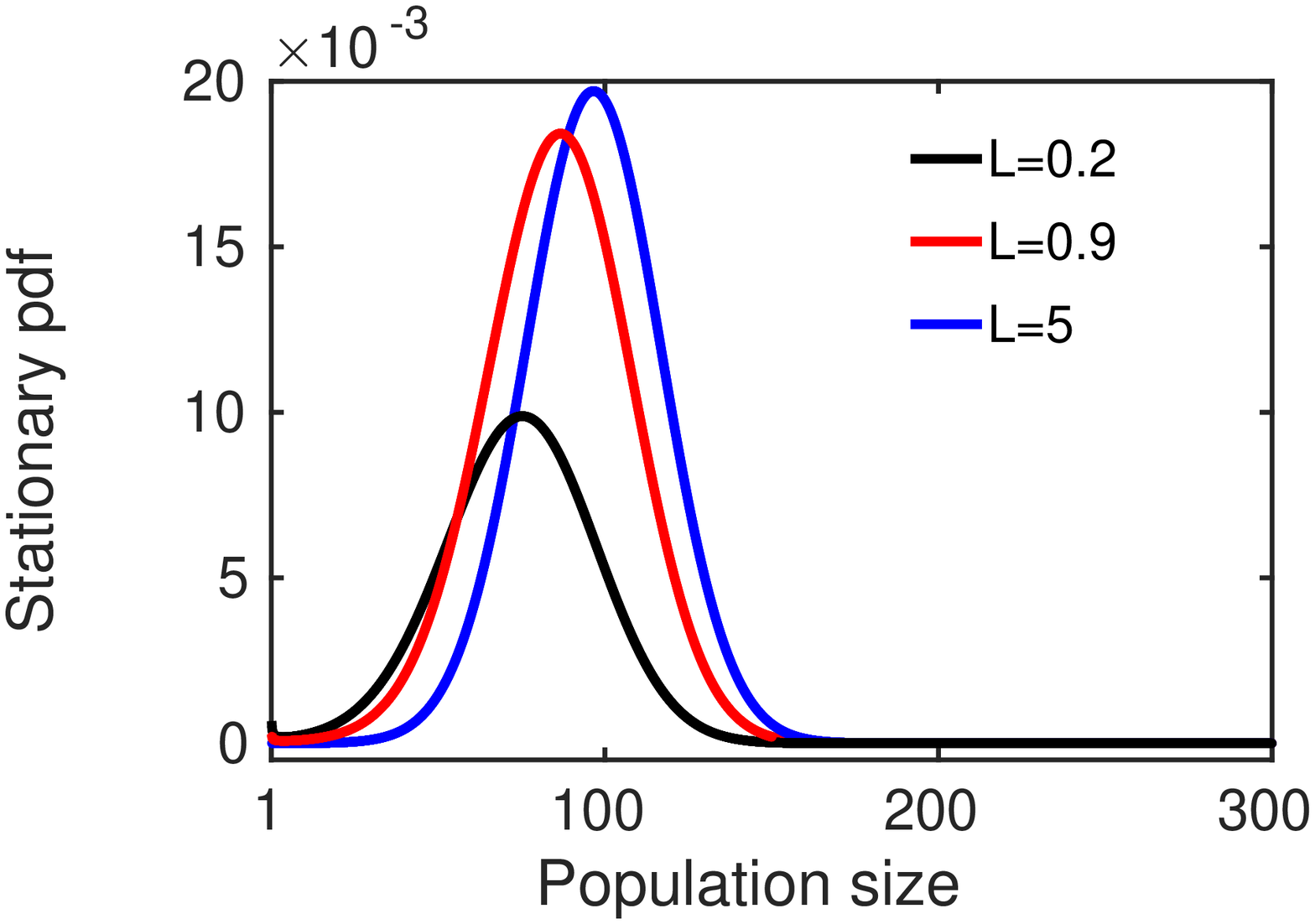}(c)
\end{center}
\caption{(a) Linear case, (b) Nonlinear case, (c)Nonlinear case.
		This figure represents the stationary probability distribution against the resource population size. The figure \ref{pdfthetalogistic}(a) depicts the stationary distribution for different $\theta$, when harvest rate is linear and the parameter values are $r$=0.5, $K$=100, $q$=1, $E$=3.5. The figure \ref{pdfthetalogistic}(b) represent the  stationary distribution for different $\theta$, when the harvest rate is nonlinear. Here the corresponding parameter values are $r$=0.5, $K$=100,  $q$=1, $a$=3.5, $E$=3.5, $L$=1.12. In the figure \ref{pdfthetalogistic}(c), we plot the stationary distribution for different $L$ in the case of nonlinear harvest rate. Here the other parameters are $r$=0.5, $\theta$=1.4, $K$=100, $q$=1, $a$=3.5, $E$=3.5.}
\label{pdfthetalogistic}
\end{figure}

\begin{figure}[H]
\begin{center}
\includegraphics[height = 80mm, width =100mm]{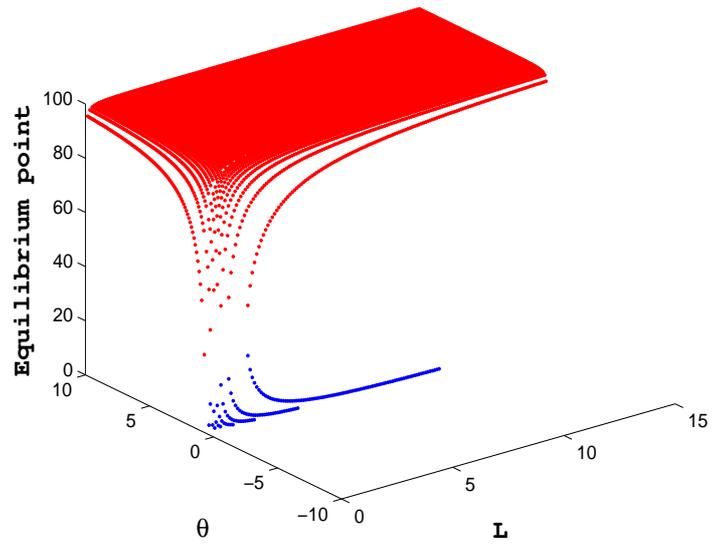}
\end{center}
\caption{The plot exhibits the equilibrium population size against the two important parameters $L$, $\theta$. The red dots indicate the stable equilibrium states, whereas the blue dots represent unstable equilibrium.}
\label{equilibrium}
\end{figure}

%\begin{figure}[H]
%\begin{center}
%\subfigure[]{\includegraphics[height = 55mm, width =75mm]{deterministic_stable(1).eps}}
%\subfigure[]{\includegraphics[height = 53mm, width = 75mm]{deterministic_unstable.eps}}
%\end{center}
%\caption{It represents the time series plot of population size under the nonlinear harvest rate in the deterministic case. Here we observe that the population always reaches near the carrying capacity and stays in the same state as time progresses. The parameter values are $r$=0.8, $K$=100, $a$=0.5, $q$=0.5, $E$=0.5. For the first figure  $L$=1.12, $\theta$=4.74, and for the second figure $L$=4.10, $\theta$=0.12.}
%\label{timeseriesplot}
%\end{figure}
%

\begin{figure}[H]
\begin{center}
\includegraphics[height=2.5 in,width=3in]{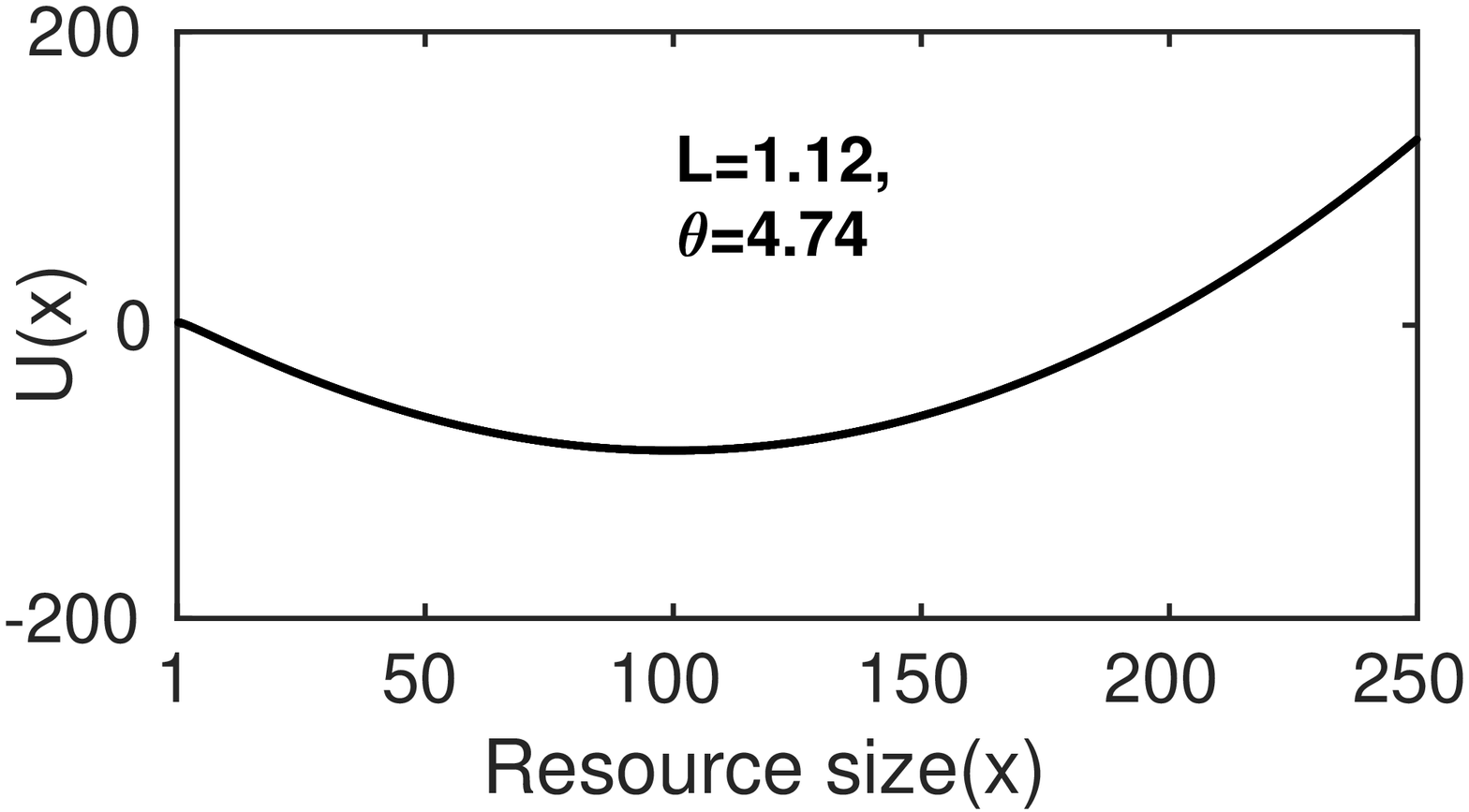}(a)
\includegraphics[height=2.5 in,width=3in]{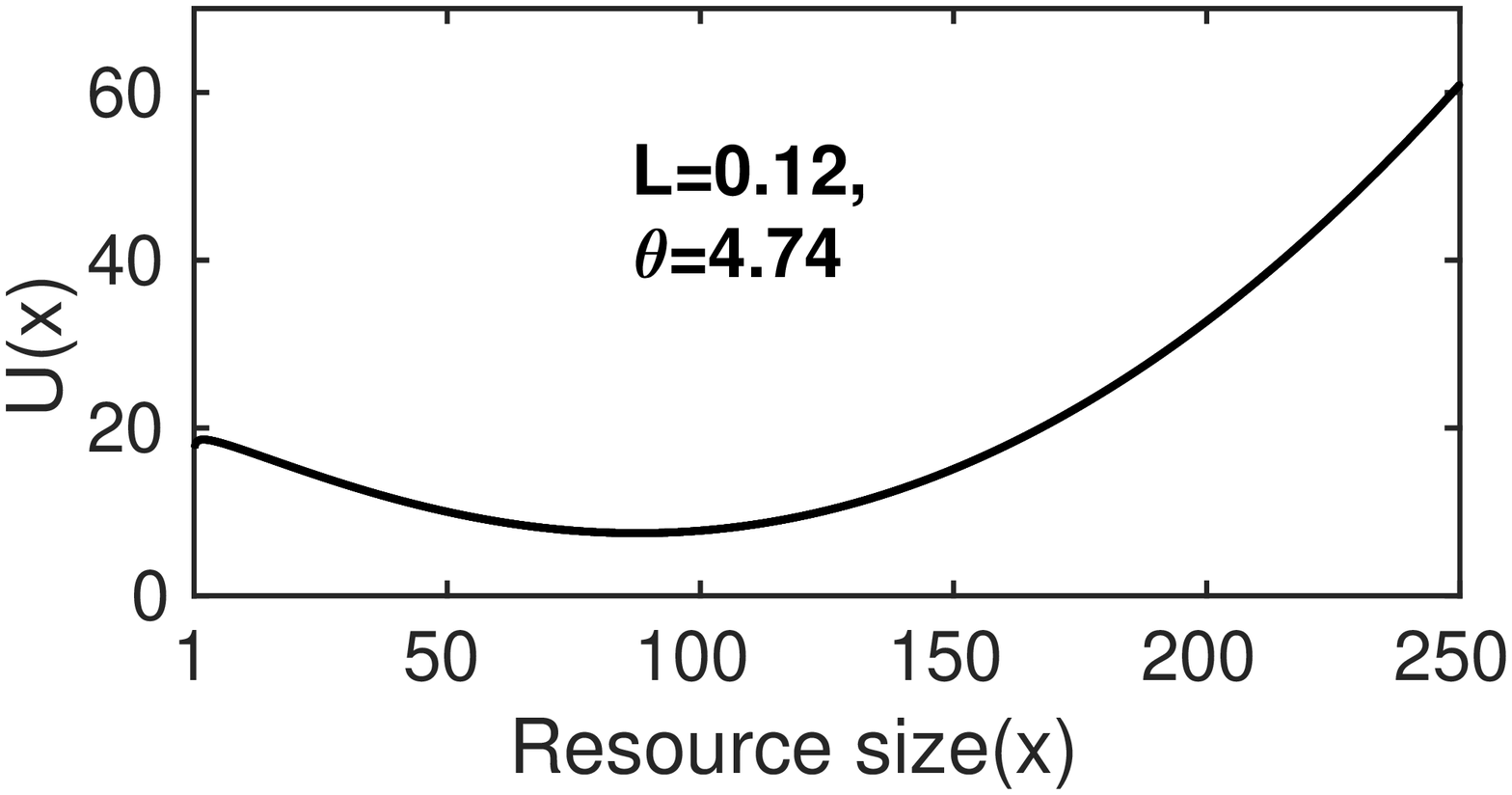}(b)
\includegraphics[height=2.5in,width=3in]{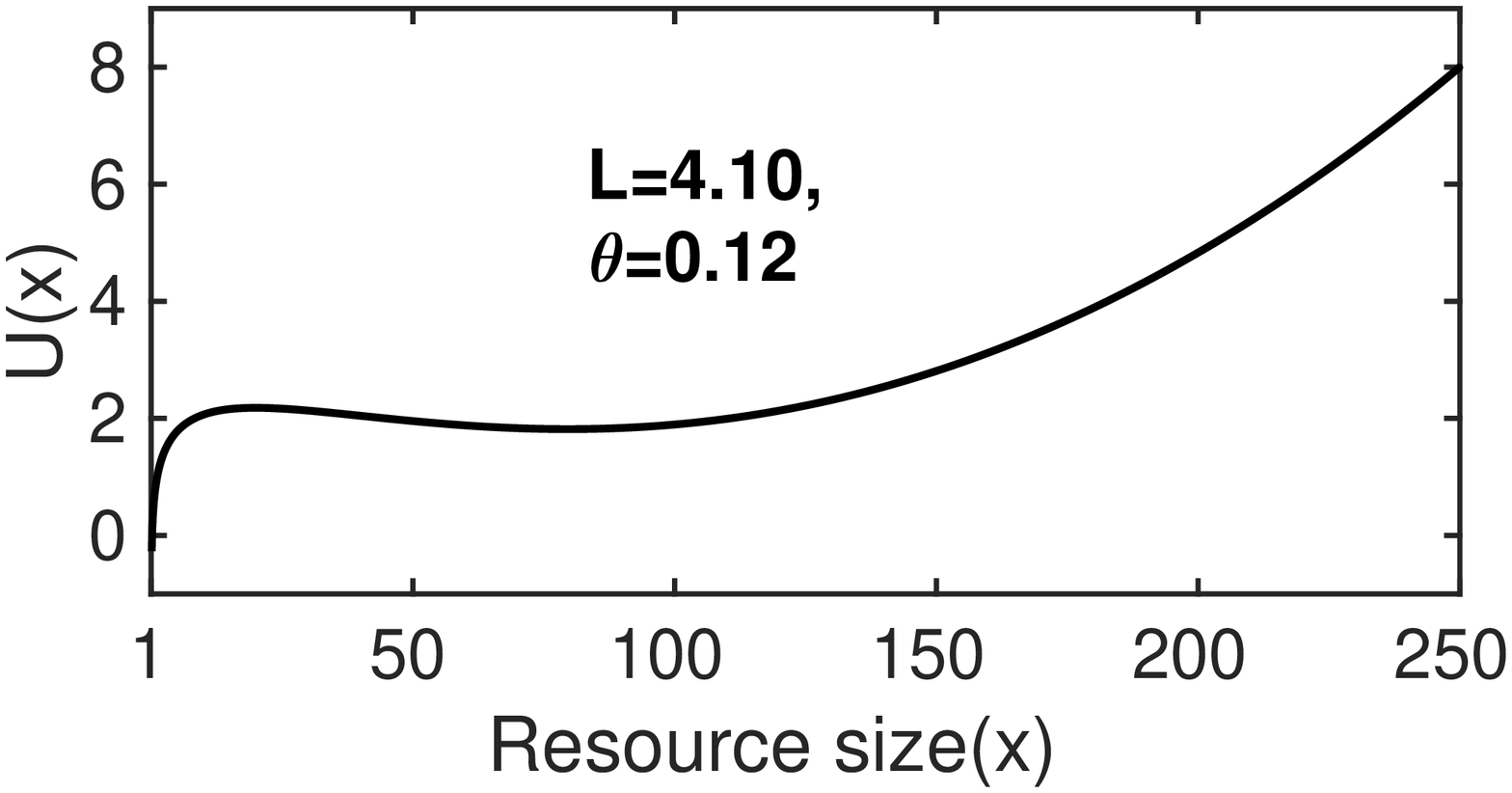}(c)
\end{center}
\caption{The panels of figures exhibit how the asymmetry in the potential landscape picture enlarges in the different pathways to catastrophic change. The corresponding parameter values are $r$=0.8, $K$=100, $a$=0.5, $q$=0.5, $E$=0.5. We consider the three points from the figure \ref{equilibrium} for different $\theta$ and $L$, i.e., we consider the points through stable equilibria to unstable equilibria and obtain the corresponding potential functions.}
\label{potentialfunction}
\end{figure}

\begin{figure}[H]
\begin{center}
\includegraphics[height=2.5 in,width=3in]{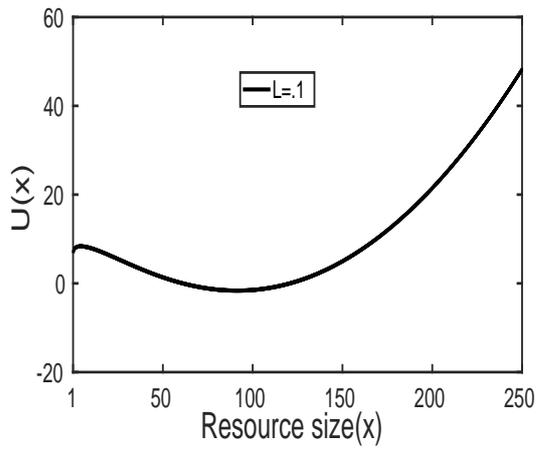}(a)
\includegraphics[height=2.5 in,width=3in]{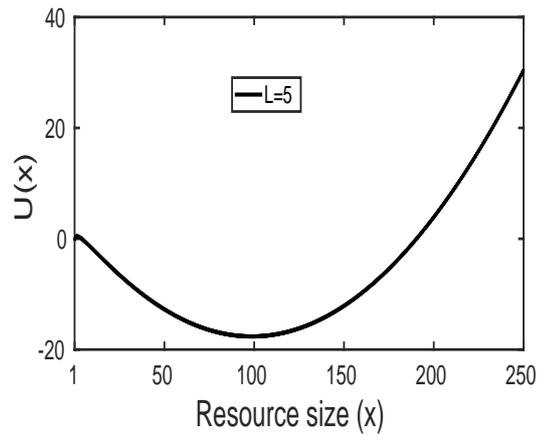}(b)
\includegraphics[height=2.5in,width=3in]{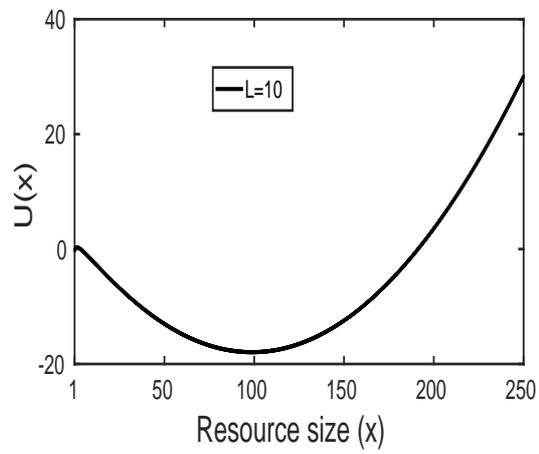}(c)
\end{center}
\caption{The panels of figures exhibit how the asymmetry in the potential landscape picture enlarges in the different pathways to catastrophic change. Hence all parameters are the same as previous, and we only vary the handling time $L$. For the fixed value of $\theta$=.74, we plot the figure for $L$=.1, 5, 10, respectively.}
\label{fig6}
\end{figure}

\begin{figure}[H]
\begin{center}
\includegraphics[height=1.8in,width=5.5in]{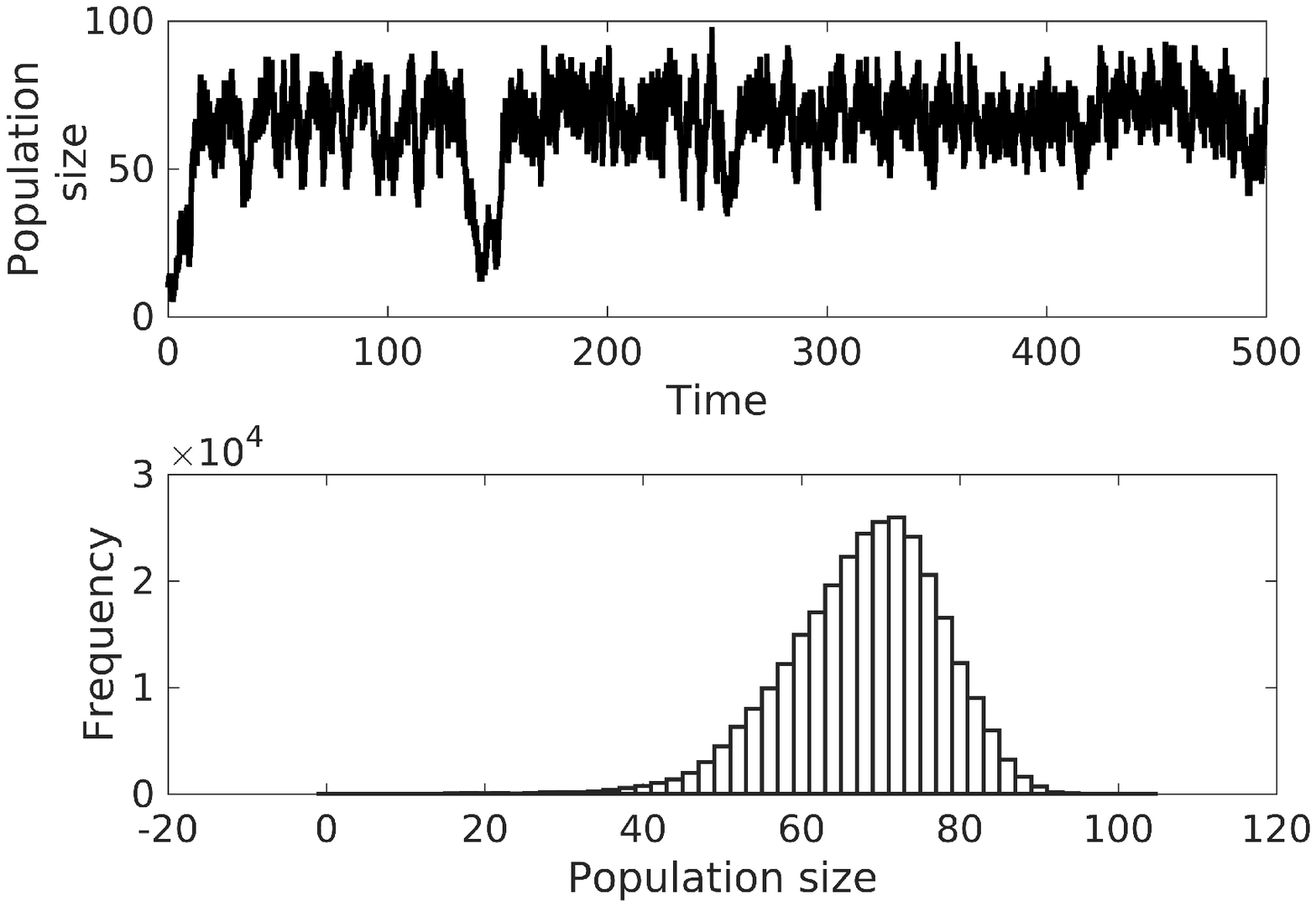}(a)
\includegraphics[height=1.8in,width=5.5in]{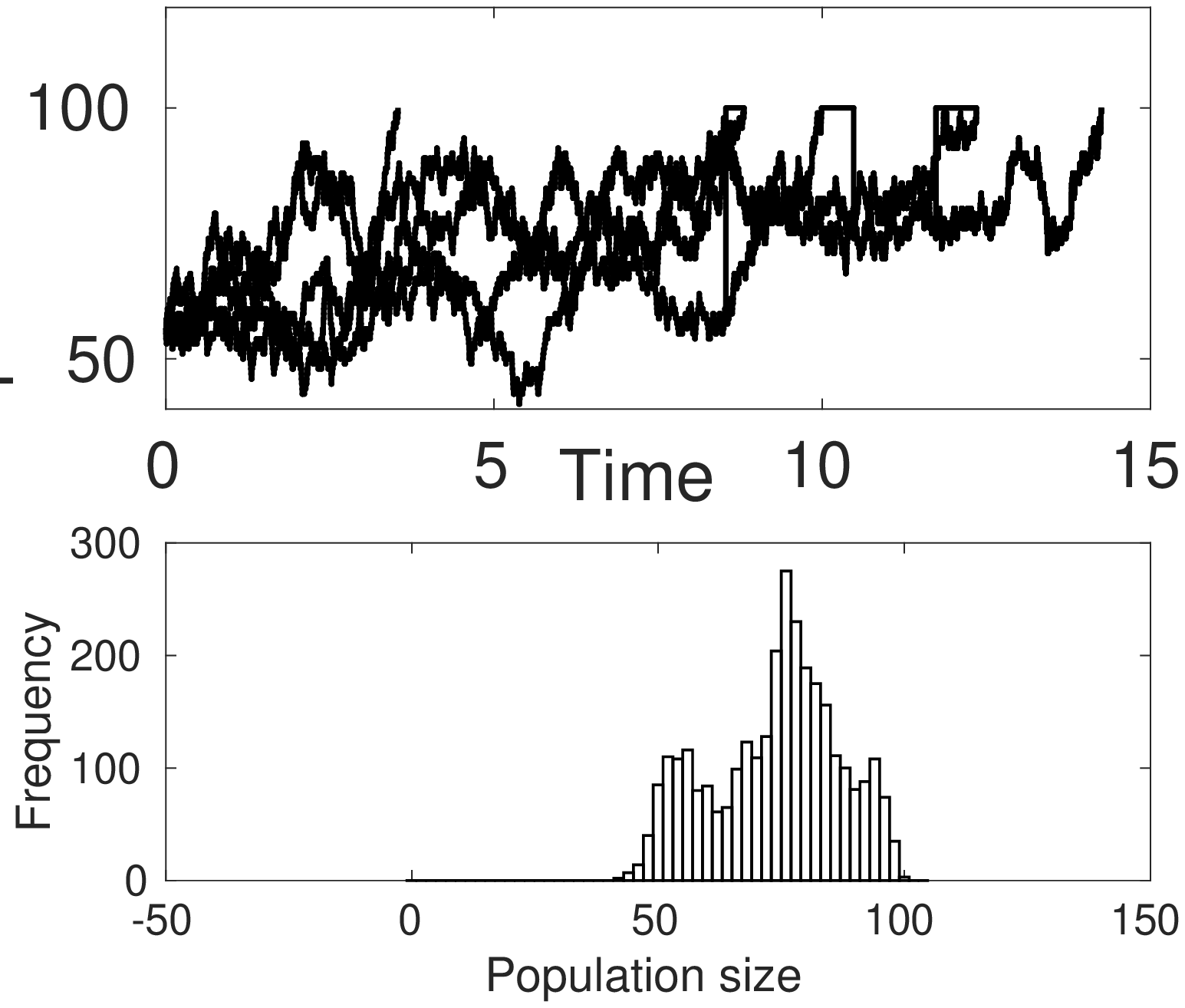}(b)
\includegraphics[height=1.8in,width=5.5in]{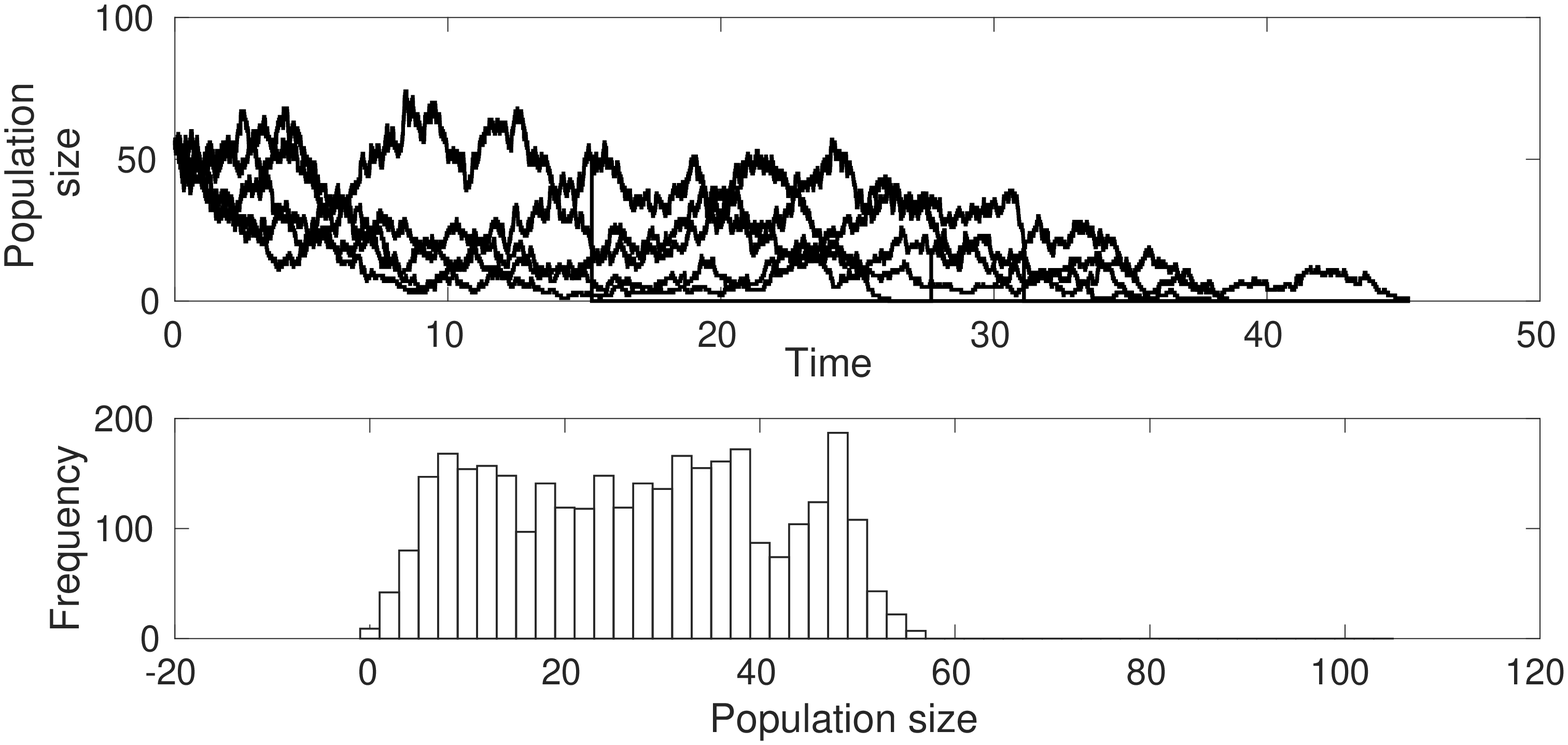}(c)
\end{center}
\caption{These figures represent numerical simulation results for the time series and its histogram when the system gradually moves to the unstable region. The increase in the asymmetry in the distribution underlies the indicator of catastrophic change. The parameter values are $r=0.2$, $K$=100, $q$=1, $a$=0.2;$e$=0.6. The  corresponding $L$ and $\theta$ are: $L$=1.12, $\theta$=4.74; $L$=0.12, $\theta$=4.74; $L$=4.10, $\theta$=0.12.}
\label{stochastic_samplepath_histogram}
\end{figure}

%%%%%%%%%%%%%%%%%%%%%%%%%%%%%%bifurcation for theta allee%%%%%%%%%%%%%%%%%%%%%%%%%%%%%%%%%%%%%

\newpage
\section{Data analysis}\label{The real data and discussion}

The extract of the study in the previous sections altogether reveals that, when a species remain in the upper stable equilibrium, we can prevent the species from going to the zero stable equilibrium by controlling the handling time in the nonlinear harvesting. This is one of the important mechanisms to control species sustainability. We validate our numerical results through population time-series data of herring fish population in two specific regions available in the GPDD. We consider the population time-series data of 16 years (1972-1987) with GPDD ID-1741 and 41 years (1951-1991) for GPDD ID-1772. For the data set GPDD ID-1741, the experiment was performed in Baltic sea areas covering 28 and 29 S area, and for the data set GPDD ID-1772, it was done in Prince Robert district, British Columbia. 

 Sau et al. \citep{sau2020extended} fitted these population size data (in Metric Tonnes) to the $\theta$-logistic and ASM considering both the linear and nonlinear harvest rates. The fitting through the RGR model using a grid search technique was performed for the species GPDD ID-1741, and ASM with nonlinear harvest rate is found to be the best-fitted model. For the species with GPDD ID-1772, they used the size modeling approach, where the raw population time-series data can be directly applied for estimation. They found that the $\theta$-logistic model with nonlinear harvesting is the best-fitted curve in this model.

In this article, we generate sample paths using the Gillespie algorithm. We have simulated the sample paths for different handling times; the remaining parameter values are obtained from Sau et al. \citep{sau2020extended}. The simulated time series for the herring fish population corresponding to GPDD ID 1741 shows that, a catastrophic regime shift may occur near the time point 1500 if the harvesting strategy is linear. The possible regime shift can be delayed effectively if nonlinear harvesting is applied. In figure \ref{samplepath_1741_greaterAlle_samplepath}(b) it is observed that a small increase in the handling time is sufficient to delay this shift up to 2000 time points.

%In this figure, we consider the initial population size above the Allee threshold. All the sample paths will converge to zero equilibrium point for the linear harvest rate, and the species will surely extinct. Whenever we increase the value of handling time, the maximum number of sample paths will go to the upper stable equilibrium point. Hence the species have less possibility to undergo regime shift and to become extinct. However, if the initial population size lies below the Allee threshold value, a similar result occurs for $L=0$. Whenever we increase the handling time, very few sample paths will go to the stable equilibrium point and for a sufficient large no of handling time, the number is almost fixed. Hence, if the initial condition lies below the Allee threshold, the species has more chance to extinct (figure \ref{samplepath_1741_lessthanAlle}).
 The sample path for GPDD id 1772 for different values of $L$ is depicted in figure \ref{samplepath_1742}. For this species, the $\theta$-logistic harvesting model with a nonlinear harvest rate is found to be the best-fitted model. As in the previous case, here also we observe that the sufficient increase in the handling time can delay the catastrophic change and thus can prevent population extinction (see figure \ref{samplepath_1742} (b).)   
%  There is an important dissimilarity of sustainability issues in the deterministic and stochastic system for this species. Although there is no bistability in the deterministic system, the stochastic framework behaves like a system with bistability. We know that if the species follow the $\theta$-logistic model, there is very less chance for species extinction \citep{sau2020extended}. If we consider moderate handling time, all sample paths will go to the stable equilibrium points i.e the chance of extinction of the species is negligible (\ref{samplepath_1742}(c)). Whereas for $L=0$, all the paths will move to the extinction state and gradually increase the value of L, some sample-path moves to the upper stable equilibrium, and the number of sample paths moving to upper stable equilibrium increases for further increase of L. This shows that the increase in L will help the species to undergo a regime shift from the extinction state to the upper stable state so that the species can sustain.
   This observation conveys to the management authority that non-linear harvesting with suitable handling time can be permissible to prevent this species from extinction.

\subsection*{Generic early warning}
Different early warning signals can be useful to detect impending regime shift directly from the time series data. In the previous section, for the GPDD ID 1741 and ID 1772, we observe the possibility of a catastrophic shift in the fish population density from an adequately high stable population density to an extinction state in the long run when handling time is low (figure \ref{samplepath_1741_greaterAlle_samplepath}, \ref{samplepath_1742}). To find the generic early warning signals, we restrict our observation on the time series data for the population size of last 500 time points for both the species since the early warning signals may fail to predict the regime shift for a long-term time series \citep{banerjee2021chemical} . In both cases, the standard deviation, autocorrelation increase where as the return rate decreases prior to the catastrophic shift (figure \ref{Earlywarning_thetalogistic}, \ref{Earlywarning_Allee}). This analysis supports the incidence of occurrence of regime-shift. Thus the generic early warning signals are found to be robust in this case.

\begin{figure}[H]
	\begin{center}
		\includegraphics[height = 55mm, width =75mm]{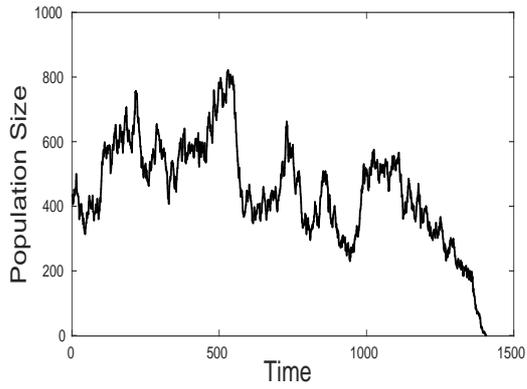}(a)
		\includegraphics[height = 53mm, width = 75mm]{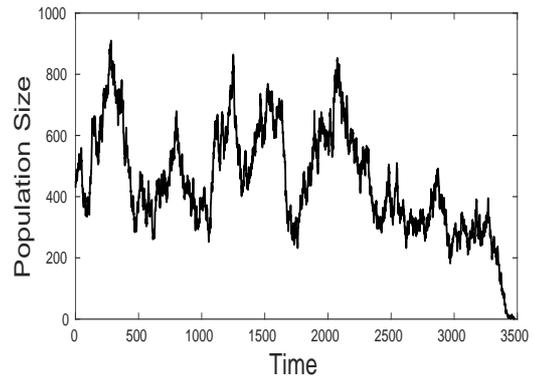}(b)
	
	\end{center}
	\caption{(a) $L=0$, $L=0.9$. Simulated sample paths for the herring population with GPDD id 1741 for different values of L. The remaining parameters are obtained from Sau et al.\citep{sau2020extended}.}
	\label{samplepath_1741_greaterAlle_samplepath}
\end{figure}

\begin{figure}[H]
	\begin{center}
		\includegraphics[height = 55mm, width =75mm]{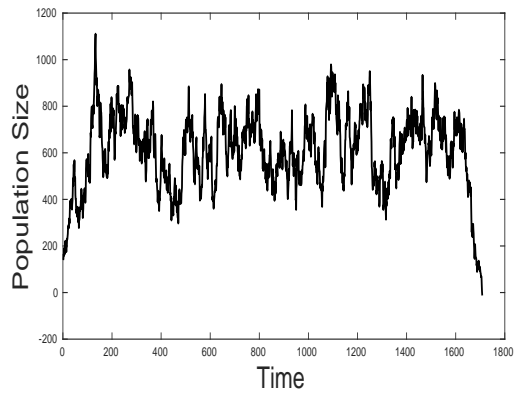} (a)
		\includegraphics[height = 53mm, width = 75mm]{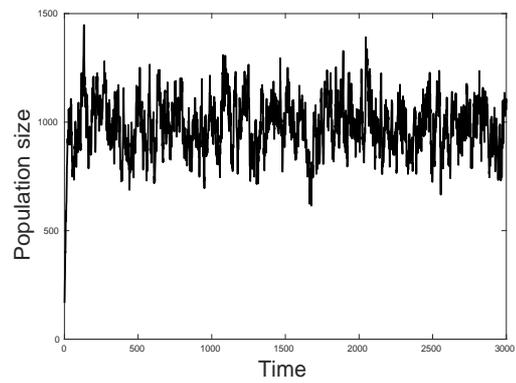} (b)
	\end{center}
	\caption{(a) linear case (L=0), (b) nonlinear case (L=3). Simulated sample paths for the herring population with GPDD id 1772 for different values of L. The remaining parameters are obtained from Sau et al.\citep{sau2020extended}.}
	\label{samplepath_1742}
\end{figure}

\begin{figure}[H]
	\begin{center}
		\includegraphics[height = 90mm, width =110mm]{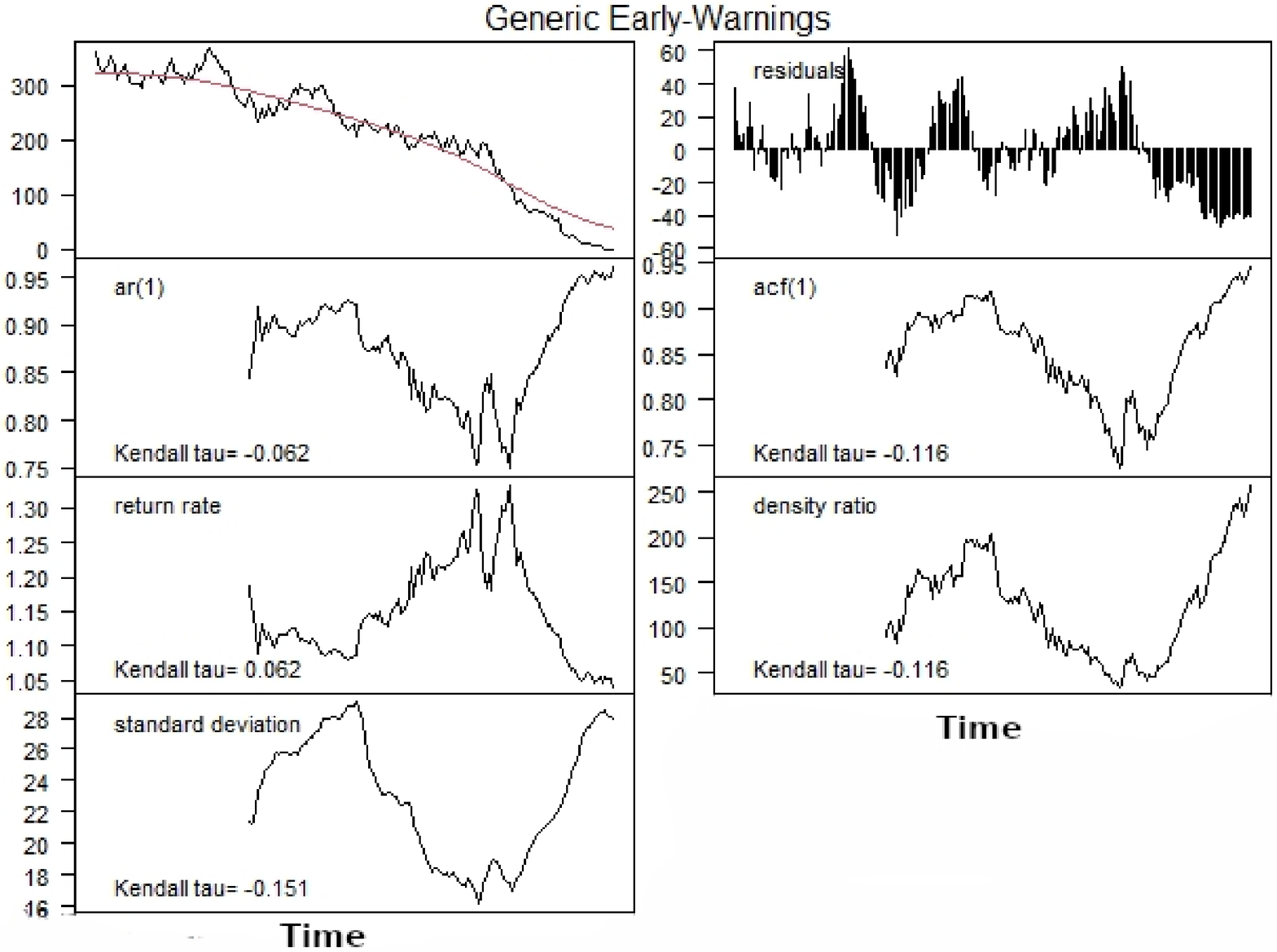}
	\end{center}
	\caption{Early warning signals for the simulated time series data corresponding to the data with GPDD ID 1741.}
	\label{Earlywarning_thetalogistic}
\end{figure}

\begin{figure}[H]
	\begin{center}
		\includegraphics[height = 90mm, width =110mm]{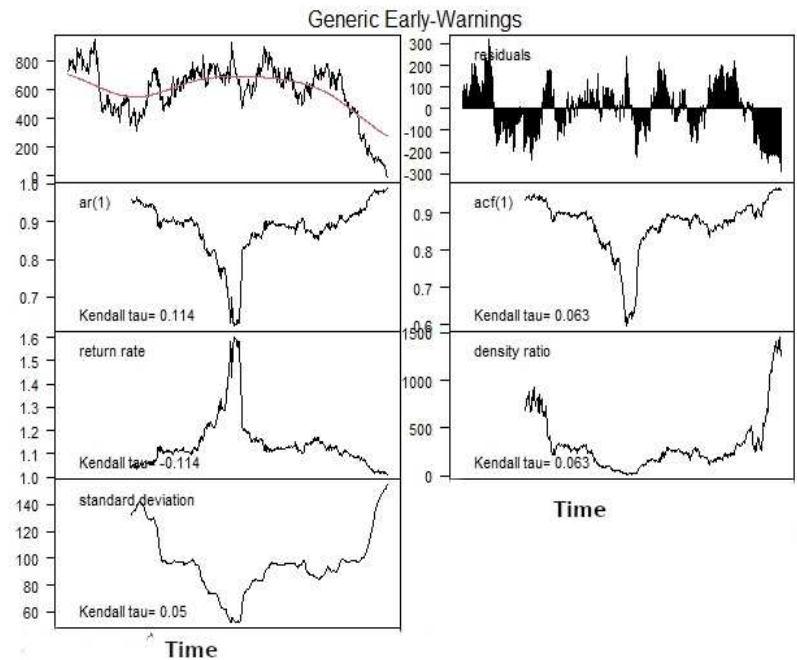}
	\end{center}
	\caption{Early warning signals for the simulated time series data corresponding to the data with GPDD ID 1772.}
	\label{Earlywarning_Allee}
\end{figure}

%\begin{figure}[H]
%	\begin{center}
%		\subfigure[linear case (L=0)]{\includegraphics[height = 55mm, width =75mm]{samlepath_1772_L0.eps}}
%		\subfigure[nonlinear case (L=3)]{\includegraphics[height = 53mm, width = 75mm]{samlepath_1772_01.eps}}
%		\subfigure[nonlinear case (L=5)]{\includegraphics[height = 55mm, width =75mm]{samlepath_1772_L02.eps}}
%	\end{center}
%	\caption{Numerical simulation results for time series of GPDD ID 1772 for nonlinear harvest rate with the $\theta$-logistic model.}
%	\label{samplepath_1742}
%\end{figure}

\section{Conclusion}\label{Discussion and conclusion}

One of the major challenging issues in ecology is maintaining an adequate abundance of different essential biological resources to prevent its extinction. Marine fisheries are salient examples of such an ecological system as the fishes are being depleted rapidly due to indiscriminate harvesting \citep{saha2013evidence}. In this context, an appropriate harvesting strategy is important to maintain the sustainability of the species. Besides this, the marine system and similar ecological systems are highly exposed to catastrophic changes, particularly when they have multiple or alternate stable states. The system with multiple stable states may be prone to extinction if one of the stable states is the extinction state.

In this article, we consider a stochastic harvesting model where the growth process of the species follows theta-logistic and Allee mechanism \cite{sau2020extended}. Stochasticity is introduced through the concept of the birth-death process. The extinction state is one of the stable states for both models. Our findings in the present work reveal that the possibility of regime shift or catastrophic changes can be restricted by the suitable choice of handling time. Note that, if the handling time ($L$) is increased, the chance of maintaining the abundance of resource population at a higher stable equilibrium point is more, and this can minimize the chance of species extinction.
%The same argument is true even if the system does not possess multiple equilibria in the stochastic environment. Further study reveals that a system may display multiple equilibria under a stochastic environment although it is missing in the deterministic framework. We have generated the potential function for understanding stability characteristics of the system. We expect a single local minimum of the potential function in the case of density regulated $\theta$-logistic model, as the model has a unique stable equilibrium at the carrying capacity of the species under consideration. But surprisingly we notice two local minima of the potential function. The additional minimum occurs around 1 i.e, at the extinction state although it is not a stable equilibrium point of the $\theta$-logistic model. This observation supports our claim that the stochastic system may differ from its deterministic counterpart qualitatively. 
%
%
%

The density regulation parameter  $(\theta)$ is one of the important structural properties of the species, and it is fixed for a particular location. In comparison, the handling time is controlled through the experimenter. The species with low theta value is more threatened, as our study suggests. Hence, in this case, it is more significant to choose proper handling time. To be more precise, in this work, we have identified the parameter region which allows bi-stability. So, for a given data set, if the estimated values of the parameters fall in this region, the biological resource manager should raise the flag to detect the possibility of a regime shift. This detection becomes urgent if one of the stable states is the extinction state. The resource management in particular the fishery management for marine or lake systems, should detect the early warning signal by applying a generic early warning toolbox or identifying the basin of attraction, or observing the asymmetry of potential function. If anyone or all of the measures indicates the possibility of regime shift, our present work can help to take proper harvesting policy by choosing suitable handling time so that the regime shift can be averted.

We establish our findings through the population time-series data of herring fish population at the Baltic sea and Prince Robert district, British Columbia. The numerical simulation suggests that there may exist a regime shift from the non-zero state to the extinction state for both species in the near future, if handling time is absent or very low. We observe that this regime shift can be avoided if handling time is taken sufficiently large.

\subsection*{Acknowledgments} We would like to acknowledge the University Grants Commission (UGC), Government of India for the financial support. Additionally we are grateful to one of our lab mates Mr. Swarnendu Banerjee for his valuable suggestions in this paper.

\bibliographystyle{plain}
\bibliography{manuscript_2ndpaper}
\clearpage

\section*{Appendix}
\subsubsection*{Appendix A (i): The condition for birth and death rates for $\theta$-logistic model}

In the model \ref{general_model_harvest}, we consider the birth rate as  
\begin{eqnarray*}
b(x)&=&(r+1)x-\frac{x^{\theta+1}}{2K^\theta}
\end{eqnarray*}
and the death rates corresponding to linear and nonlinear harvest rate are respectively
\begin{eqnarray*}
 d(x)&= & x+\frac{x^{\theta+1}}{2K^\theta} + qex \\
 d(x)&=& x+\frac{x^{\theta+1}}{2K^\theta} + \frac{qex}{aE+Lx}
\end{eqnarray*}
 
\begin{thm}
 If $\left(r-\frac{q}{a}\right)>0$ then $\exists$ $M$, $N$ such that $b(x)>d(x)$ $\forall$ $x$ $\in$ $(0,M)$ and $b(x)<d(x)$ $\forall$ $x$ $\in$ $(M,N)$
\end{thm}

\begin{proof}
 $b(x)-d(x)= rx\left(1-\left(\frac{x}{K}\right)^\theta\right)-\frac{qEx}{aE+Lx} =f(x)$ (say). 
  Now $f(0)=0$ and $f'(0)= \left(r-\frac{q}{a}\right)>0$ (from our assumption).
  Again $f(x)<0$ $\forall$ $x\geq K$. Hence $\exists$ $M\in (0,K)$ such that $f(M)=0$ and $f(x)>0$ $\forall$ $x\in (0,M)$ and $f(x)<0$ $\forall$ $x > M$.\\
   We have,\\
 $b(x)= x\left[(r+1)-\frac{1}{2}\left(\frac{x}{K}\right)^\theta\right]>0$
 $\Rightarrow x<K\left[2(r+1)\right]^\frac{1}{\theta}$\\
 We assign $N= K\left[2(r+1)\right]^\frac{1}{\theta}$. Clearly $N>0$ as $r>0$ and $b(x)>0$ if $x$ $\in$ $(0,N)$.
\end{proof}

\noindent\textbf{Remark:} Similarly proceeding as above we can find $M$ and $N$ for proportional harvesting also. 

\subsubsection*{Appendix (ii): The condition for birth rate and death rate for the ASM}
 The ASM is\\
$\frac{dx}{dt}= rx\left(\frac{x}{K}-\frac{A}{K}\right)\left(1-\left(\frac{x}{K}\right)^\theta\right)$\\
$= r\frac{A+K}{K^{\theta+1}}x^{\theta+1}\left(1-\frac{x}{A+x}\right)-\frac{rA}{K^\theta}x^\theta $
$= \frac{rx^2}{K^{\theta+1}}\left(K^\theta+Ax^{\theta-1}\right)\left[1-\frac{x^\theta}{K^\theta+Ax^{\theta-1}}\right] +\frac{rax}{K}$.\\
Introducing the nonlinear harvesting phenomena, the model will be \\
$\frac{dx}{dt}= rx\left(\frac{x}{K}-\frac{A}{K}\right)\left(1-\left(\frac{x}{K}\right)^\theta\right) + \frac{qEx}{aE+Lx}$\\
In this case we can write the model equivalent to $b(x)-d(x)$,\\
where, $b(x) =\frac{rx^2}{K^{\theta+1}}\left(K^\theta+Ax^{\theta-1}\right)\left[1-\frac{x^\theta}{K^\theta+Ax^{\theta-1}}\right]$ and~ $d(x)= \frac{rax}{K} + \frac{qEx}{aE+Lx} $\\
\noindent Here we observe that $ b(x)$ and $d(x)$ are always positive in the interval $(0,N]$. To justify $b(x)<d(x)$, we have,~
$rx\left(\frac{x}{K}-\frac{A}{K}\right)\left(1-\left(\frac{x}{K}\right)^\theta\right)-\frac{qEx}{aE+Lx} < 0$\\
For the above inequality, the sufficient condition is, 
\begin{eqnarray*}
&& \frac{r}{K}(x-A)-\frac{qE}{aE+Lx}<0  ~~[\mbox {Putting $(\frac{x}{K})^\theta = 0$}]\\
&& \Rightarrow (x-A)(aE+Lx)-\frac{qEK}{r}<0\\
&& \Rightarrow Lx^2-(LA-aE)-\left(AaE+\frac{qEK}{r}\right)<0\\
&& \Rightarrow x= \frac{(LA-aE)\pm\sqrt{(LA-aE)^2+4L\left(AaE+\frac{qEK}{r}\right)}}{2L}\\
&& \Rightarrow (x-\alpha)(x-\beta)<0
 \end{eqnarray*}
where, 
\begin{eqnarray*}
 \alpha = \frac{(LA-aE)+\sqrt{(LA-aE)^2+4L\left(AaE+\frac{qEK}{r}\right)}}{2L},~~~
 \beta = \frac{(LA-aE)-\sqrt{(LA-aE)^2+4L\left(AaE+\frac{qEK}{r}\right)}}{2L}
\end{eqnarray*}
 Let us consider $\beta = -m$ (say),$m>0$. Therefore $(x-\beta)>0$,
 $\Rightarrow (x-\alpha)<0$, $\Rightarrow x<\alpha$. We can assume $A^* = \alpha$
% Similarly the sufficient condition for existence of $M$ is given by
\begin{equation}
 r\left(\frac{A^*}{K}-\frac{A}{K}\right)\left(1-\left(\frac{M}{K}\right)^\theta\right)> \frac{qE}{aE+LA^*}
\end{equation}

\noindent\textbf{Remark:} We can apply the same process to find $M$ and $N$ for proportional harvesting also.

%Data:
%\begin{figure}[H]
%	\begin{center}
%		\subfigure[L=0]{\includegraphics[height = 55mm, width =75mm]{gpdd_1751_greaterthanallee_L_low.eps}}
%		\subfigure[L=.9]{\includegraphics[height = 53mm, width = 75mm]{gpdd_1751_greaterthanallee_L_9.eps}}
%		\subfigure[L=19]{\includegraphics[height = 55mm, width =75mm]{gpdd_1751_greaterthanallee_L_19.eps}}
%	\end{center}
%	\caption{Sample path of GPDD ID 1741 for nonlinear harvest rate with the Allee model. The initial size of the population is greater than Allee threshold. }
%	\label{samplepath_1741_greaterAlle}
%\end{figure}
%
%
%\begin{figure}[H]
%	\begin{center}
%		\subfigure[L=0]{\includegraphics[height = 55mm, width =75mm]{GPDDID_1741_bellowalle_samplepath_L01.eps}}
%		\subfigure[L=.07]{\includegraphics[height = 53mm, width = 75mm]{GPDDID_1741_bellowalle_samplepath_L07.eps}}
%		\subfigure[L=20]{\includegraphics[height = 55mm, width =75mm]{GPDDID_1741_bellowalle_samplepath_L20.eps}}
%	\end{center}
%	\caption{Numerical simulation results for time series of GPDD ID 1741 for nonlinear harvest rate with the Allee model. The initial size of the population is below the Allee threshold.}
%	\label{samplepath_1741_lessthanAlle}
%\end{figure}

\end{document}